\documentclass[10pt,twocolumn,twoside]{IEEEtran}

\pdfoutput=1
\usepackage{etex}
\usepackage[letterpaper,top=0.75in,left=0.75in,right=0.75in,bottom=0.75in]{geometry}
\reserveinserts{28}

\newcommand{\subparagraph}{}
\usepackage{titlesec}
\usepackage{caption}
\usepackage[nodisplayskipstretch]{setspace}

\usepackage{setspace}
\usepackage{amsmath, amsfonts, amssymb, mathtools, url, graphicx, balance}
\usepackage[sort,nocompress]{cite}
\usepackage{cases}
\usepackage[usenames, dvipsnames]{color}

\usepackage{mathrsfs,upgreek}

\usepackage{amsthm}
\usepackage{cleveref}
\crefname{section}{\mytextsection}{\mytextsection\mytextsection}
\Crefname{section}{Sec.}{§§}
\crefname{subsection}{\mytextsection}{\mytextsection\mytextsection}

\newtheorem{theorem}{Theorem}
\newtheorem{lemma}{Lemma}

\newtheorem{corollary}{Corollary}
\newtheorem{remark}{Remark}
\newtheorem{proposition}{Proposition}
\newtheorem{definition}{Definition}
\newtheorem{assumption}{Assumption}

\newtheorem*{formulation*}{}

\newtheorem{example}{Example}

\crefname{theorem}{Thm.}{Theorems}
\crefname{proposition}{Prop.}{Propositions}
\crefname{lemma}{Lemma}{Lemmas}
\crefname{algorithm}{Algorithm}{Algorithms}
\crefname{figure}{Fig.}{Figures}

\DeclareGraphicsExtensions{.pdf,.jpeg,.png}
\makeatletter
\newtheoremstyle{named}{}{}{\itshape}{}{\bfseries}{.}{.5em}{\thmnote{#3's }#1}
\theoremstyle{named}

\newtheoremstyle{mynamed}{}{}{\itshape}{}{\bfseries}{.}{.5em}{#1 \thmnote{A }}
\theoremstyle{named}

\allowdisplaybreaks
\usepackage[table,x11names]{xcolor}

\usepackage{pgfplots}
\usepackage{subfig}
\usepackage{caption}

\usepackage{fancybox}
\usepackage{tikz}
\usetikzlibrary{patterns}
\usepackage{multirow}
\usepackage{hhline}
\usepackage{mathtools}
\usepackage{chngcntr}
\usepackage{pgfplots}
\pgfplotsset{compat=newest}
\counterwithout{equation}{section} 
\counterwithout{figure}{section} 
\usetikzlibrary{arrows,shapes,snakes,automata,backgrounds,petri}
\usepackage{pgfplots}
\usepackage{booktabs}
\usepackage{algpseudocode}
\usepackage{mathabx}
\usepackage{floatrow}
\usepackage{algorithm,algorithmicx,algpseudocode}
\algrenewcommand{\algorithmiccomment}[1]{\hskip3em$\slash\slash$ #1}

 \usepackage{mathrsfs}
\usepackage{etex,etoolbox}

\usepackage{etex,etoolbox}
\usepackage{xfrac}

\newcommand{\Abb}{\mathbb{A}}
\newcommand{\Pbb}{\mathbb{P}}
\newcommand{\Ebb}{\mathbb{E}}
\newcommand{\Nbb}{\mathbb{N}}

\newcommand{\Mcal}{\mathcal{M}}
\newcommand{\Vcal}{\mathcal{V}}
\newcommand{\Ecal}{\mathcal{E}}
\newcommand{\Gcal}{\mathcal{G}}
\newcommand{\Pcal}{\mathcal{P}}
\newcommand{\Hcal}{\mathcal{H}}
\newcommand{\Lcal}{\mathcal{L}}

\newcommand{\Ncal}{\mathcal{N}}

\def \Leps {\Lambda_{\epsilon}}

\def \Aeps {A_{\epsilon}}

\def \mytextsection {\textsection\hspace{-0.0cm}}

 \usepackage{relsize}

\usepackage{chngcntr}
\usepackage{epstopdf}
\counterwithout{subsubsection}{subsection}

\titleformat{\subsubsection}[runin]
  {\normalfont\normalsize\bfseries\itshape}{\thesubsubsection)}{2pt}{}[:]
\makeatletter
\@addtoreset{subsubsection}{subsection}
\makeatother

\title{Asymptotic Network Robustness}
\author{T. Sarkar, M. Roozbehani, M. A. Dahleh
\thanks{T. Sarkar, M. Roozbehani, M. A. Dahleh are with the EECS
Department, Massachusetts Institute of Technology, Cambridge, MA 02139 USA (e-mail: \texttt{tsarkar, mardavij, dahleh}@mit.edu).}
}

\begin{document}
\maketitle
\setlength{\abovedisplayskip}{0.1cm}
\setlength{\belowdisplayskip}{0.1cm}
\begin{abstract}
This paper examines the dependence of performance measures on network size with a focus on large networks. We develop a framework where it is shown that poor performance can be attributed to dimension--dependent scaling of network energy. Drawing from previous work, we show that such scaling in undirected networks can be attributed to the proximity of network spectrum to unity, or distance to instability. However, such a simple characterization does not exist for the case of directed networks. In fact, we show that any arbitrary scaling can be achieved for a fixed distance to instability. The results here demonstrate that it is always favorable, in terms of performance scaling, to balance a large network. This justifies a popular observation that undirected networks generally perform better. We show the relationship between network energy and performance measures, such as output shortfall probability or centrality, that are used in economic or financial networks. The strong connection between them explains why a network topology behaves qualitatively similar under different performance measures. Our results suggest that there is a need to study performance degradation in large networks that focus on topological dependence and transcend application specific analyses.
\end{abstract}


\section{Introduction}
\label{sec:introduction}

{Networks have been successfully used to model many real world phenomena such as input--output sectoral networks in economics~\cite{alireza_daron}, vehicular platoons in transportation~\cite{fiedler_nonzero} and opinion dynamics in social sciences~\cite{olshevsky2013degree}}. These networks are characterized by the nature of interconnections, \textit{or topology}, and dimension, which is typically very large. Due to the preponderance of such networks, there is extensive literature on suitable performance measures and their behavior on different topologies. However, the idea of a ``suitable'' performance measure, as we discuss below, varies across applications. A major objective of this paper is to understand the underlying relationship between existing performance measures for large networks. 

There is a substantial body of interdisciplinary research to find appropriate performance measures for large interconnected networks with stable dynamics and understand the underlying causes of dimension dependent scaling of these measures (See~\cite{daron_aggreg},~\cite{robust_consensus},~\cite{fiedler_nonzero},~\cite{h2_norm_vol} and~\cite{market_fragile} for example). A commonly used performance measure for design of optimal controllers is the $\Hcal_2$--norm. Specifically, Lin et. al. in~\cite{lin2012optimal} study the design of structured controllers for vehicular platoons with minimum $\Hcal_2$--norm. It is observed there that optimal structured controllers with asymmetric feedback exhibit lower scaling (of $\Hcal_2$--norm) with the number of vehicles than its symmetric counterpart and performs better in that sense. On the other hand,~\cite{robust_consensus},~\cite{h2_norm_vol} study the dependence of $\Hcal_2$--norm on structural properties of the network such as edge weights and underlying graph spectrum. Huang et. al. study $\Hcal_2$--norm based volatility measures in undirected (symmetric) and stable network, where the volatility measure can be represented completely by the network spectrum. The scaling of $\Hcal_2$--norm for different network topologies such as ring, star and cycle along with the effect of nodal degree perturbation and the consequent problem of critical edge identification are also studied there. Another important performance measure for vehicular platoon networks is the $\Hcal_{\infty}$--norm. The notion of harmonic instability, \textit{i.e.}, exponential scaling of $\mathcal{H}_{\infty}$--norm in vehicular platoons is studied in~\cite{fiedler_nonzero}. Herman et. al. study network topologies where a shock on the leader vehicle of a stable platoon network is magnified exponentially in its size. Due to this exponential scaling, control of such a network becomes increasingly problematic as its dimension grows. {Briefly, $\Hcal_2, \Hcal_{\infty}$--norms or volatility measures are different manifestations of energy of a network, which is described in the future sections.}

A seemingly different performance measure than the system norms discussed before is studied in~\cite{daron_aggreg},~\cite{alireza_daron}. Acemoglu et. al. introduce a measure called tail risk to assess performance in economic production networks. Tail risk measures the probability of output falling below a certain threshold in response to an input shock. A key feature of tail risk is that it can manifest only in networks with large dimension, and hence captures the idea of performance degradation when network size increases perfectly. A major conclusion of their work is that some amount of balancedness is needed in a large network to exhibit no tail risk.

In control systems theory robust stability, with no structural constraints on the feedback, is well understood (See~\cite{morari1989robust} for details). Robust stability entails the study of controllers that provide internal stability to all networks (\textit{or plants}) in a given set. This is particularly useful in the case when there is some uncertainty about the network, such as unknown edge weights etc. The notion of ``network robustness'' that we discuss here is different from traditional robust stability. We focus on large networks with stable dynamics where we study the behavior of performance measures such as the $\Hcal_2$--norm or tail risk. Specifically, we are interested in the scaling of these measures as a function of network dimension. We study the underlying causes of this scaling and find conditions under which it can be attributed to the proximity of network spectrum to unity. This introduces the idea of asymptotic instability which is characterized by the spectral radius approaching unity as network size increases, for a given topology. We will study the dependence of performance scaling on asymptotic instability and show when it does not completely explain such scaling. We will observe that excessive scaling in performance measures are closely related to notions of criticality, centrality and robustness in graphs, that are well studied in algebraic graph theory. However, there the graphs considered are ``static'' (no associated dynamics) or at equilibrium conditions. Real life networks, on the other hand, are characterized by their associated topology along with their complex interconnections and system dynamics. Therefore, a simultaneous consideration of control and graph theoretic tools is necessary for the design of robust networks, where performance measure scales gracefully with network size. 

Although there is a myriad of techniques to assess performance in large networks it is not clear if there exists some connection between them. For example it is shown by Acemoglu et. al. that balanced networks exhibit no tail risk; on the other hand Lin et. al. show that symmetric networks (a subclass of balanced networks) do not perform optimally. The dichotomy of these results is surprising since the underlying dynamics are similar. A major contribution of this work is that we provide a general framework to analyze these seemingly disparate results. We also find optimal (in an appropriate sense) topologies, where performance does not degrade rapidly as network size increases. 

The presentation of the paper is as follows, we introduce some mathematical notation in Section~\ref{sec:notation}. In Section~\ref{sec:Preliminaries} we present the mathematical preliminaries required to understand the paper. We formalize the notion of large networks used in this paper and describe the network topologies that will be discussed throughout the paper. We also define network robustness in large networks here and show its relation to network energy. Following this, in Section~\ref{robust_prelim} we characterize robustness for undirected (symmetric) networks and prove that such a characterization does not extend to directed networks. Through these observations, we show that undirected topologies exhibit the best performance for a fixed spectral radius. We establish connections between network robustness and different performance measures used in economics, transportation and finance in Section~\ref{sec:robustnessMeasures}. The main results of the paper can be found in Sections~\ref{robust_prelim},~\ref{sec:robustnessMeasures}. Finally, we conclude in Section~\ref{sec:conclusion}. All proofs are in the appendix.

\section{Mathematical Notation}
\label{sec:notation}

\textbf{Matrix Theory:} A vector $v \in \mathbb{R}^{n \times 1}$ is of the form $\lbrack v_1, \ldots, v_n \rbrack^T$, where $v_i$ denotes the $i^{th}$ element, unless specified otherwise. The vector $\textbf{1}$ is the all $1$s vector of appropriate dimension; to specify the dimension we sometimes refer to it as $\textbf{1}_n$, where it is a $n \times 1$ vector. Similarly, for a $m \times n$ matrix, $A$, we refer to it as $A_{m \times n}$ when we want to specify dimension. We denote $A$ that is a $n \times n$ matrix as $A_n$ for short hand. {For a matrix, $A$, we denote by $\rho(A)$ its spectral radius. Additionally, we have $\rho_i(A) \geq \rho_{i+1}(A)$ where $\rho_i(A) = |\lambda_i(A)|$ the $i^{th}$--eigenvalue of $A$}. Thus, $\rho_1(A) = \rho(A)$. We have a similar notation for the singular values of $A$ denoted by $\sigma_i(A)$. $I$ is the identity matrix of appropriate dimension. 

The $\Lcal_p$ norm of a matrix, $A$, is given by 
	$$||A||_{p} = \sup_{v} ||Av||_p/||v||_p$$
$||A||_S$ is the Schatten norm, \textit{i.e.}, $||A||_S = \sum_{i=1}^n \sigma_i(A)$. For positive semidefinite (psd) matrices the Schatten norm equals the trace of the matrix. {The symbol $\succeq$ denotes the Loewner order between two Hermitian matrices $A, B$, \textit{i.e.}, if $A \succeq B$ or $A \succeq_{PSD} B$ then $A-B$ is a psd matrix, similarly if $A \preceq B$ or $A \preceq_{PSD} B$ then $B-A$ is a psd matrix.}

\textbf{Order Notation:} For functions, $f(\cdot), g(\cdot)$, we have $f(n) = O(g(n))$, when there exist constants $C, n_0$ such that $f(n) \leq C g(n)$ for all $n \in \Nbb > n_0$. Further, if $f(n) = O(g(n))$, then $g(n) = \Omega(f(n))$. For functions $g(\cdot), h(\cdot)$, we have $g(n) = \Theta(h(n))$ when there exist constants $C_1, C_2, n_1$ such that $C_1 h(n) \leq g(n) \leq C_2 h(n)$ for all $n \in \Nbb > n_1$. Finally, for functions $h_1(\cdot), h_2(\cdot)$, we have $h_1(n) = o(h_2(n))$ when $\lim_{n \rightarrow \infty} |h_1(n)/h_2(n)| = 0$. 

\textbf{Graph Theory:} A graph is the tuple $\Gcal = (\Vcal_{\Gcal}, \Ecal_{\Gcal}, w_{\Gcal})$, where $\Vcal_{\Gcal} = \{v_1, v_2, \ldots, v_n\}$ represents the set of nodes and $\Ecal_{\Gcal} \subseteq \Vcal_{\Gcal} \times \Vcal_{\Gcal}$ represents the set of edges or communication links. An edge or link from node $i$ to node $j$ is denoted by $e\lbrack i, j \rbrack = (v_i, v_j) \in \Ecal_{\Gcal}$, and $w_{\Gcal}: \Ecal_{\Gcal} \rightarrow \mathbb{R} $. Denote by $\Abb_{\Gcal}$ the adjacency matrix of $\Gcal$. A graph, $\Gcal$, is symmetric or undirected if $w_{\Gcal}(v_i, v_j) = w_{\Gcal}(v_i, v_j)$ for all $1 \leq i, j \leq |\Vcal_{\Gcal}|$. $\Gcal$ is induced by a matrix, $A_{n \times n}$ if $\Vcal_{\Gcal} = \{1, \ldots, n\}$, and $(i, j) \in \Ecal_{\Gcal}$ if $\lbrack A \rbrack_{ij} \neq 0$, and $w_{\Gcal}(i, j) = \lbrack A \rbrack_{ij}$. { The adjacency matrix of a graph $\Gcal$ is denoted by $\Abb$. $\Abb$ is a $|\Vcal_{\Gcal}| \times |\Vcal_{\Gcal}|$ matrix with the following property: $[\Abb]_{ij}=1$ if there is an edge $i \rightarrow j$ else it is $0$.}

{ \textbf{Probability Theory:} For a random variable $X$, we define 
\begin{align*}
\mu(X) &= \Ebb[X] \\
\sigma(X) &= {\Ebb[(X-\mu(X))^2]}^{1/2}
\end{align*}
$\Phi(\cdot)$ is the cumulative distribution function of Gaussian distribution and is given by
\[
\Phi(x) = \frac{1}{\sqrt{2\pi}} \int_{-\infty}^{x} e^{-t^2/2} dt
\] 
We sometimes refer to the probability density function and cumulative distribution function of a random variable as pdf and cdf respectively. For a random variable $X$, the $\tau$--tail probability is given by 
$$\Pbb\Bigg(\frac{|X - \mu(X)|}{\sigma(X)} > \tau\Bigg)$$
}

\textbf{Miscellaneous: }Denote by $\Pcal_d$ is the family of polynomials with degree $\leq d \in \Nbb$.

\section{Preliminaries}
\label{sec:Preliminaries}
\begin{definition}
	\label{network}	
	A network is a graph $\Gcal = (\Vcal_{\Gcal}, \Ecal_{\Gcal}, w_{\Gcal})$, where $\Vcal_{\Gcal} = \{1, 2, \ldots, n\}$, and for each node, $i \in V_{\Gcal}$, there is an associated dynamical behavior, $i \rightarrow x_i(\cdot)$. Further, $w_{\Gcal}(i, j) = \lbrack A \rbrack_{ij}$ for all $1 \leq i, j \leq n$. 
\end{definition}
In this work, we will focus on networks with linear dynamics given by Eq.~\eqref{DT_LTI} described below
\begin{equation}
\label{DT_LTI}
x(k+1) = A\; x(k) + \omega \; \delta(0, k), \hspace{2mm} k \in \{0, 1, 2, \ldots \}
\end{equation}
Here $x(k) = \lbrack x_1(k), \ldots, x_n(k) \rbrack^{T}$ is the vector of state variables. $A$ is the $n \times n$ state transition matrix. $\delta(0, k)$ is the Kronecker delta function, with $\delta(0, 0) = 1$ and $\delta(0, k) = 0 \hspace{2mm} \forall \, k \neq 0$ and $\omega = \lbrack \omega_1, \ldots, \omega_n \rbrack^{T}$ is exogenous to the system. We further assume that $x(0) = 0$. Additionally, we have the following assumptions
{\begin{assumption}
	\label{stability}
	The dynamical behavior of the network is governed by Eq.~\eqref{DT_LTI}. Further, the state transition matrix $A$ is Schur stable, \textit{i.e.}, $\rho(A) < 1$.
\end{assumption}
Assumption~\ref{stability} is an assumption on the network dynamics and its stability that will be enforced on all networks in this work, unless otherwise stated. We will call $1 - \rho(A)$ as the distance to instability of $A$.
\begin{definition}
	\label{undirected}
A network is undirected if $A$ is symmetric, otherwise it is a directed network.
\end{definition}}
\begin{definition}
	\label{shock}
A signal, $w(k)$, is a shock if $w(k) = \textbf{0}$ for all $k > 0$ but $w(0) \neq \textbf{0}$.
\end{definition}
Then $w(k) = \omega \delta(0, k)$ is a shock.
\begin{remark}
	\label{network_rep}
	 Every network, in this paper, can be uniquely characterized by its state transition matrix, or network matrix, $A$.
\end{remark}

\begin{figure}[h]
	\centering
 	\includegraphics[width=0.8\textwidth]{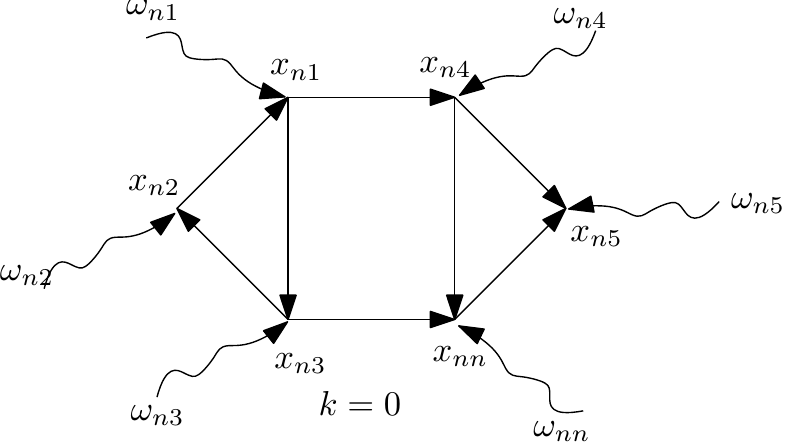}
	\caption{Network with noise on every node, $k = 0$}
	\label{noise}
\end{figure}

The key feature of our work is that we address large networks. We denote a large network by ``a sequence of networks with network matrices $\{A_{n}\}_{n=1}^{\infty}$, where the topology of each network in the sequence is fixed but network dimension grows successively''. This treatment is along the lines of graph limits as discussed in~\cite{lovasz_graph_limits}. For example, to denote a large star network we define a sequence of star networks $\{A_n\}_{n=1}^{\infty}$ with dimensions $ n \times n = 1 \times 1, 2 \times 2, 3 \times 3, \ldots$. In this example we have 
	\begin{align*}
	A_n=\begin{bmatrix}
	0  & \frac{1}{n-1} & \frac{1}{n-1} & \dots & \frac{1}{n-1} \\
	1       & 0 & 0 & \dots & 0 \\
	\vdots & \vdots & \vdots & \ddots & \vdots \\
	1       & 0 & 0 & \dots & 0
	\end{bmatrix}
	\end{align*}
From here on, whenever we use ``sequence of networks'' we will mean a ``large'' network.
\\
A mathematical quantity that will play an important role in the discussion of performance scaling is the (identity) Gramian that we define below(See~\cite{chen1998linear} for a detailed exposition).

\begin{definition}
	\label{gramian}
	The (identity) gramian of a Schur stable matrix, $A$, is given by $P(A)$, or, $P$, where $P = A^T P A + I$. 
\end{definition}

Throughout this paper we consider some commonly encountered networks. We summarize them below. \vspace{3mm}\\
\textbf{Networks with degree normalization}\\
These networks are of the form $A_n = \gamma D^{-1}_n \Abb_n$ where $D_n$ is the degree matrix, $\Abb_n$ is the adjacency matrix of $A_n$ and $0 <  \gamma < 1$.
\begin{figure}[h]
	\centering
		\includegraphics[width=0.8\textwidth]{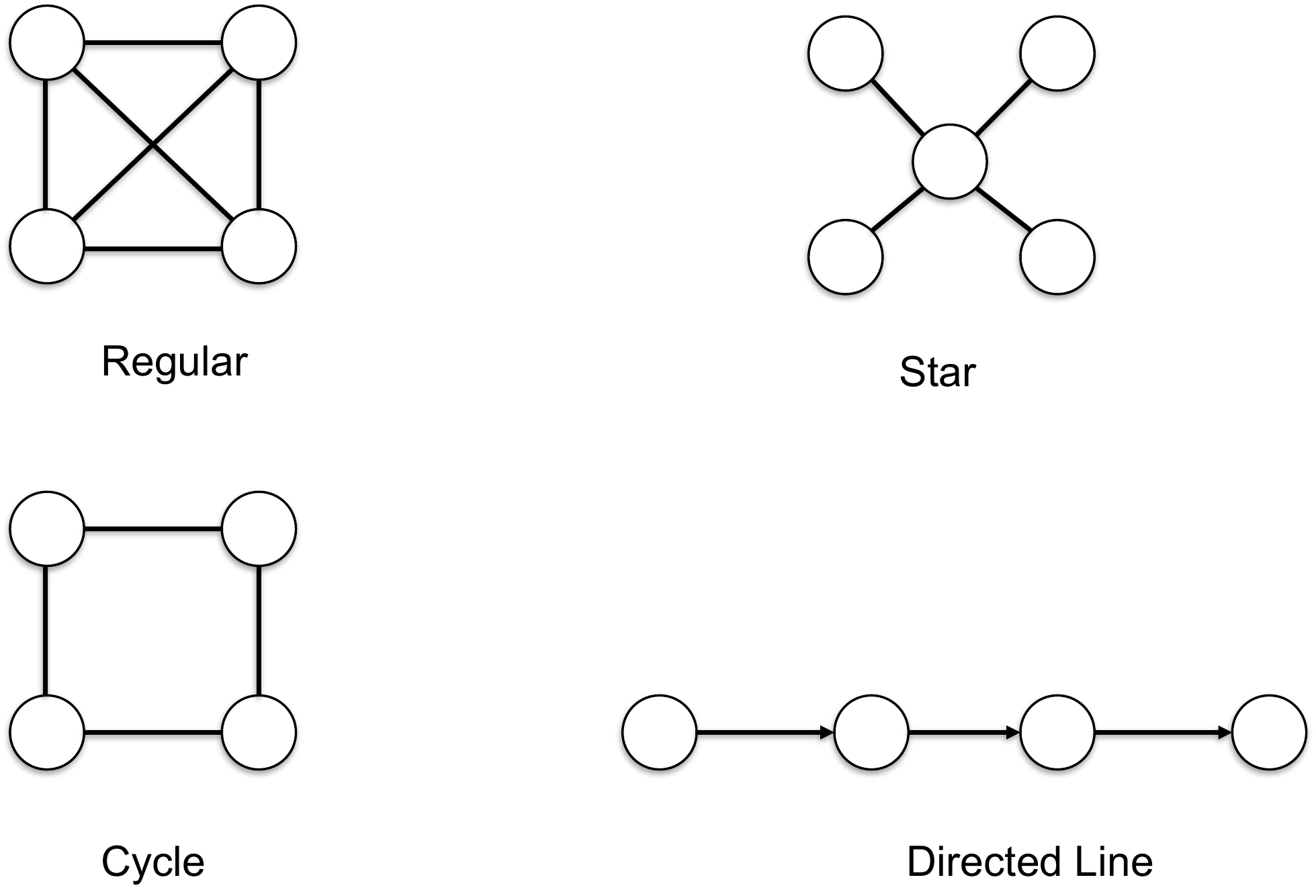}	
	\caption{Different Network Topologies}
	\label{macro_risks}
\end{figure}
Edges with arrow are directed edges. We will refer to the regular, star, cycle, directed line degree networks as $R_n, S_n, C_n, DL_n$ respectively.
{	\begin{align*}
	R_n=\frac{\gamma}{n-1}(\textbf{1} \textbf{1}^T - I)
	\end{align*}
\begin{align*}
[DL_n]_{i, j} &=   
\begin{cases}
1 & \text{if $j = i+1$} \\
0 & \text{otherwise}
\end{cases}
\end{align*}
\begin{align*}
[S_n]_{i, j} &=   
\begin{cases}
\gamma & \text{if $i \neq j$ and $j=1$} \\
\frac{\gamma}{n-1} & \text{if $i \neq j$ and $i=1$}\\
0 & \text{otherwise}
\end{cases}
\end{align*}	
\begin{align*}
[C_n]_{i, j} &=   
\begin{cases}
\gamma/2 & \text{if $i = j+1$} \\
\gamma/2 &  \text{if $j = i+1$} \\
\gamma/2 &  \text{if $(i, j) = (1, n)$ or $(i, j) = (n, 1)$} \\
0 & \text{otherwise}
\end{cases}
\end{align*}} \vspace{3mm}\\
\textbf{Network with loops}\\
\begin{figure}[h]
	\centering
	\includegraphics[width=7cm]{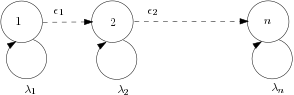}
	\caption{Network $\textbf{VP}_n$} \label{self_loops}
\end{figure}
The topology in Fig.~\ref{self_loops} appears frequently in transportation networks as vehicular platoons. The values of $\lambda_i, \epsilon_i$ are chosen to ensure Schur stability of the network.
\begin{align*}
[\textbf{VP}_n]_{i, j} &=   
\begin{cases}
\lambda_i & \text{if $i = j$} \\
\epsilon_i &  \text{if $j = i+1$} \\
0 & \text{otherwise}
\end{cases}
\end{align*}\\
\\
{\textbf{Random Networks}\\
Our results can be extended to random discrete time LTI systems. We consider the Wigner ensemble defined as follows
\begin{definition}
	\label{wigner_ens}
A Wigner matrix, $W_n(\sigma)$, is a matrix that has the following properties:
\begin{itemize}
	\item For all $1 \leq i < j \leq n$, $W_n(\sigma)_{ij}$ is i.i.d. and $W_n(\sigma)_{ij} \sim \Ncal(0, \sigma^2)$.
	\item  $[W_n(\sigma)]_{ij} = [W_n(\sigma)]_{ji}$
	\item $\{[W_n(\sigma)]_{ii}\}_{i=1}^n$ are i.i.d. with mean $0$ and variance $\sigma^2$. For all $l \neq j$, $[W_n(\sigma)]_{ii}$ is independent of $[W_n(\sigma)]_{lj}$.  
\end{itemize}
A Wigner ensemble is a discrete time LTI system where $A_n = \frac{W_n(\sigma)}{\sqrt{n}}$ where $\sigma < \frac{1}{2}$.
\end{definition}
}
The performance measures we are concerned with are closely related to the ``network energy'' that we describe below
{\begin{definition}
	\label{h2norm}
	The $\Hcal_2$--norm of a network $A$ is given by 
	\[
	\Hcal_2(A) = \text{trace}(P(A))
	\] 
\end{definition}
$\Hcal_2$--norm is a common system norm that will occur frequently in our discussion to measure performance of a network. In fact, $\Hcal_2$--norm measures the cumulative amplification of a shock due to network effects, \textit{i.e.}, network energy. Here we present an interpretation of network energy. At time $k = 0$, each node $i$ is hit with a shock, $\omega_{i}$, then we are interested in the following questions - 
\begin{itemize}
	\item If $\omega = \lbrack \omega_{1}, \ldots, \omega_{n} \rbrack^T$ is a deterministic shock, with $||\omega||_2 \leq 1$, then what is maximum effect of the shock on the network ?
	\item If $\omega$ is a random shock, then what is the effect of the shock on the network, on average ?
\end{itemize} 
We measure the energy of the network by the quantity, $E_{\infty} = \sum_{k=0}^{\infty} x^{T}(k) x(k)$, and without loss of generality we will assume $x(0) = \textbf{0}$.
\begin{definition} 
	\label{max_norm}
	Given a network, $A$, with a deterministic input shock, $\omega$, max norm, $\Mcal(A)$, is given by: 
	\begin{equation}
	\label{defn1}
	\Mcal(A) = \sup_{||\omega||_2 = 1} E_{\infty}
	\end{equation}
	We will refer to $\Mcal(A)$ as the maximum disruption energy of  $A$.
\end{definition}

\begin{definition}
	\label{average_norm}
	Given a network, $A$, with a  a random shock, $\omega$, average norm (per node), $\Ecal(A)$, is defined as the following: 
	\begin{equation}
	\label{defn2}
	\Ecal(A) = (1/n)\Ebb_{\omega} \lbrack E_{\infty} \rbrack
	\end{equation}
	Here $\Ebb[\omega] = \textbf{0}, \Ebb[\omega \omega^T] = I$. We will refer to $\Ecal(A)$ as the average disruption energy of  $A$.
\end{definition}
{
\begin{proposition}
	\label{norm_sing2}
	The max norm, $\Mcal(A)$, is $\sigma_{1}(P(A))$, and the average norm, $\Ecal(A)$, is $(1/n)\Hcal_2(A)$.
\end{proposition}}
$\Mcal(A)$ and $\Ecal(A)$ represent two different aspects of network robustness. When there is a deterministic shock incident to the network, $\Mcal(A)$ denotes the maximum energy that can propagate through the network. On the other hand if it is a random shock, then $\Ecal(A)$ denotes expected energy that propagates through the network (due to each node). However, $\Mcal(A), \Ecal(A)$ are closely related through the Gramian $P(A)$: which itself gives a complete energy profile for a network.

Whenever $\rho(A) \geq 1$ it is well known that network energy increases in an unbounded fashion(See~\cite{chen1998linear}). To simply observe the effect of interconnections on network energy, by factoring out the distance to instability \textit{i.e.} $1 - \rho(A)$, we propose the scaled $\Hcal_2$--norm.
{
	\begin{definition}
	\label{scaled_h2norm}
	The scaled $\Hcal_2$--norm of a network $A$, denoted by $\Hcal_2^S(A)$, is 
	\[
	\Hcal^S_2(A) = (1 - \rho(A))\text{trace}(P(A))
	\] 
\end{definition}
}
Using the $\Hcal_2$--norm as a performance measure (or robustness measure) we formalize the notion of ``graceful scaling''. 
\begin{definition} 
	\label{resil}
	A large network, $A_n$, is robust (in an asymptotic sense) if we have:
	\begin{itemize}
		\item Network matrix, $A_n$, is stable for each $n$
		\item $\Hcal_2(A_n) = O(p(n))$
	\end{itemize}
	Here $p(\cdot) \in \Pcal_d$ for some $d \in \Nbb$. Fragility is the lack of robustness in the sense of super-polynomial or exponential scaling of $\text{trace}(P(A_n))$ for $\{A_n\}$. Further, let $A^1, A^2$ be two networks and $\Hcal_2(A_n^1)  = \Theta(p_1(n)), \Hcal_2(A_n^2)  = \Theta(p_2(n))$. Then
	\begin{itemize}
		\item If $p_1(n) = o(p_2(n))$ then $A_n^1$ performs better than $A_n^2$ in $\Hcal_2$--norm.
		\item If $p_1(n) = \Theta(p_2(n))$ then $A_n^1$ performs as well as $A_n^2$.
		\item If $p_2(n) = o(p_1(n))$ then $A_n^1$ performs worse than $A_n^2$.
	\end{itemize}
	
\end{definition}

\begin{definition}[Spectral Balancing]
	\label{balance_networks}
	Given a large network, $A_n$, and any $0 < \epsilon < 1$, we define its $\epsilon$--balanced version as
	\[
	A^{\epsilon}_n = (1-\epsilon)A_n + \epsilon U_n \Gamma_n U_n^{T}
	\]
	where $\rho(A_n) = \gamma < 1$, $\Gamma_n$ is a diagonal matrix with $\rho(\Gamma_n) \leq \gamma$ and spectral decomposition of $P(A_n)$ given by $P(A_n) = U_n D_n U_n^{T}$. 
\end{definition}

\begin{remark}
	\label{balance_rmk}
	It is easy to observe that for an $\epsilon$--balanced version we have that 
	\[
	\max(a^{\epsilon}_{ji}, a^{\epsilon}_{ij}) - \min(a^{\epsilon}_{ji}, a^{\epsilon}_{ij}) \leq 	\max(a_{ji}, a_{ij})-\min(a_{ji}, a_{ij})
	\]
	for every pair $i, j$. Balancing makes a network more undirected.
\end{remark}}
Spectral balancing will be useful in understanding what network topologies perform better and will be discussed in the future sections.

\section{Robustness in Large Networks}
\label{robust_prelim}
We analyze network energy of an undirected network and show that it is completely characterized by its spectral radius.
\begin{proposition}
	\label{norm_sing}
	For a large undirected network, $A_n$,  we have
	\begin{align*}
	\Mcal(A_n) &= \frac{1}{1 -\rho^2(A_n)}\\
	\Ecal(A_n) &= O(1/(1- \rho(A_n)))
	\end{align*}
\end{proposition}
Proposition~\ref{norm_sing} is an extension of Proposition $1$ in~\cite{h2_norm_vol}. According to this, the scaling in the ``network volatility'' or $\Hcal_2$--norm of a large network is limited by the proximity of its spectral radius to unity. For completeness we present the $\Hcal_2$--norm of the networks introduced before.
\begin{proposition}
	\label{network_robustness}
	For the topologies in Fig.~\ref{macro_risks}
	\begin{center}
		\begin{tabular}{ |c|c|c| } 
			\hline
			Network & $\Mcal(\cdot)$ & $\Ecal(\cdot)$ \\ 
			\hline
			$S_n$ & $\Theta(n)$ & $\Theta(1)$ \\ 
			\hline
			$R_n, C_n$ & $\Theta(1)$ & $\Theta(1)$ \\ 
			\hline
			$DL_n$ &  $\Theta(n)$ &  $\Theta(n)$ \\ 
			\hline
			Wigner Ensemble & $\Theta(1)$ & $\Theta(1)$ \\ 
			\hline
		\end{tabular}
	\end{center} 
	The result for Wigner ensemble is true with probability $1$.
\end{proposition}
{The proof of this proposition and numerical verification can be found in Section~\ref{robustness_proof}.} It turns out that the maximum disruption in $S_n$ is more than in $R_n$. This happens because one unit of energy given to the central node is transferred to every other node at the next time instant (an amplification of $n-1$). On the other hand in $R_n$ only $1/(n-1)$ energy is transferred to every other node. However, both have the similar disruption on average because the probability of picking the central node is $1/n$ (assume uniform) and as a result 
\begin{equation*}
\Ecal(S_n) \approx (1/n) (n-1) + \sum_{j=2}^n (1/n) = \Theta(1)
\end{equation*}

{Another surprising observation is the high network energy for $DL_n$ despite the spectral radius being zero. One possible explanation is the unidirectional nature of the network where energy is transferred from one node to the other as is.} The natural question to ask is if there exists a characterization of network energy based on spectral radius for directed networks. 
\begin{proposition}
	\label{platoon_growth}
	For the large network $\textbf{VP}_n$(in Fig.~\ref{self_loops}) with $\epsilon_i=1$ and $\delta_1 < \lambda_i < \delta_2$ for all $i$ and some $0 < \delta_1, \delta_2 < 1$, we have that the $\Hcal_2(\textbf{VP}_n)$ is $\Omega(\exp{(cn)})$ and $\rho(\textbf{VP}_n) \leq \delta_2$.
\end{proposition}
Based on Proposition~\ref{platoon_growth}, we see that poor scaling of performance measures, in general networks, cannot be attributed to their spectra. In fact, under certain conditions they might turn out to be independent.

In the next section we discuss the applications of ideas developed here. We present a well studied model of a platoon network. We show that balancing cannot increase network energy -- which gives us a direction to improve performance in networks. Second, we present the case of an economic network. Although the performance measure seems quite different to network energy they are closely related by the Gramian.

\section{Connections to Performance Measures}
\label{sec:robustnessMeasures}
\subsection{Robustness as Network Energy}
\label{energy_network}
We revisit the discussion of $\Hcal_2$--norm scaling in different networks (directed and undirected) from the previous section. We observe that, whenever possible, making a network more undirected always results in better performance scaling. 
%
\begin{theorem}
	\label{symmetrizing}
For any large network $A_n$ with $\rho(A_n) = \gamma$ we have that 
\[
\Ecal(A^{\epsilon}_n)= O(\Ecal(A_n)), \,\Mcal(A^{\epsilon}_n) = O(\Mcal(A_n))
\]
Here $A^{\epsilon}_n$ is $\epsilon$--balanced version of $A_n$ for any fixed $0 < \epsilon < 1$.
\end{theorem}	
{The implications of Theorem~\ref{symmetrizing} are multifold. It guarantees that spectral balancing does not increase energy in an order sense. The intuition behind balancing, \textit{i.e.}, making the network ``more symmetric'', follows from Proposition~\ref{norm_sing} where we show that symmetric networks have the best possible scaling of network energy in network dimension. Indeed there are examples where an asymmetric controller is the most optimal (See Lin et. al.~\cite{lin2012optimal}) and balancing does not strictly reduce energy. However, designing such a controller is very challenging. Instead, Theorem~\ref{symmetrizing} suggests that searching over balancing operations provides an easier alternative that may reduce network energy. In fact in Fig.~\ref{symmetrizing_random} we show that in many cases spectral balancing strictly reduces network energy.}
\vspace{3mm}\\
\textbf{Scaling in Network  Controllers} \\
Here we present more concrete examples from transportation networks to show how robustness scaling manifests in controllers for vehicular platoons. We consider the vehicular platoon network discussed in~\cite{lin2012optimal}. The closed loop dynamical system (Eqns.(SS) and (VP2) in~\cite{lin2012optimal}) is of the form (discretized single integrator version)
\begin{align}
\label{optimal_controller}
x(k+1) - x(k) &= -K x(k) + d(k) \\
x_i(k+1) - x_i(k)&= k_{i, i-1}(x_{i-1}(k) -  x_i(k)) \nonumber\\
&+ k_{i, i+1}(x_{i+1}(k) -  x_i(k)) \nonumber
\end{align}
where $x_i(k)$ is the relative position error for vehicle $i$. Here $K$ is of the form 
\begin{align*}
K = \begin{bmatrix}
k_{1, 1} + k_{1, 2}  & -k_{1, 2} &   \\
-k_{2, 1} & k_{2, 1} + k_{2, 3} &  \ddots &   \\
&\ddots &  \ddots & -k_{n-1, n} \\
&      & -k_{n, n-1} & k_{n, n} + k_{n, n-1}
\end{bmatrix}
\end{align*}
Assume first that $K$ is symmetric, \textit{i.e.}, 
\begin{align*}
K = \begin{bmatrix}
k_{1} + k_{ 2}  & -k_{ 2} &   \\
-k_{2} & k_{2} + k_{3} &  \ddots &   \\
&\ddots &  \ddots & -k_{n-1} \\
&      & -k_{n-1} & k_{n} + k_{ n-1}
\end{bmatrix}
\end{align*}
Additionally, $\rho(I-K) < 1$ for stability. The structure of the feedback is given in Fig.~\ref{platoon_feedback}. $x_i$ denotes the position of vehicle $i$.

\begin{figure}[h]
	\centering
	\includegraphics[width=0.9\textwidth]{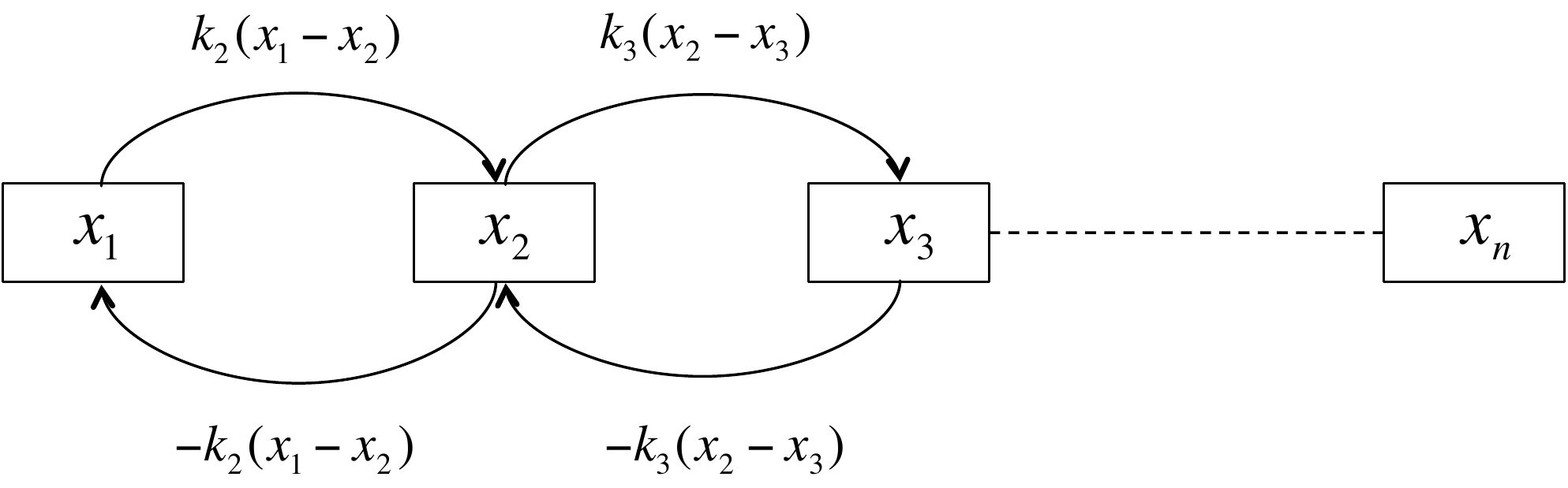}
	\caption{Feedback structure for symmetric platoon}
	\label{platoon_feedback}
\end{figure}
Now, the optimal symmetric controller design is given by
\begin{align}
\label{opt_controller}
K_{\text{sym}} = \arg \text{minimize}_{K} \hspace{0.5cm} &\text{trace}(P) \nonumber\\
\text{where} \hspace{0.5cm}  (I-K) P (I-K) &+ I = P
\end{align} 

\begin{proposition}
	\label{optimal_scaling}
	For the optimal symmetric controller in Eq.~\eqref{optimal_controller} we have 
		\begin{align*}
		\Hcal_2(I - K_{\text{sym}}) &= \Omega(n^2)
	\end{align*}
\end{proposition}
Although the number of nodes in the network are $n$, we observe that even the optimal controller gives a $\Omega(n)$ per vehicle scaling in the robustness measure. Such a performance degradation might be harmful when strict control is required. Next, we consider an asymmetric controller studied in~\cite{fiedler_nonzero} given by $K^{\epsilon}$ ($0 \leq \epsilon_i \leq 1$)
\begin{align*}
K^{\epsilon} = \begin{bmatrix}
k_1 + k_2 \epsilon_2  & -k_2 \epsilon_2 &   \\
-k_2 & k_2 + k_3 \epsilon_3 &  \ddots &   \\
&\ddots &  \ddots & -k_n \epsilon_n\\
&      & -k_n & k_n + k_{n+1} \epsilon_{n+1}
\end{bmatrix}
\end{align*} 
\vspace{1.5mm}\\
The solving for the optimal controller
\begin{align}
\label{opt_controller_asm}
K_{\text{asm}} = \arg \text{minimize}_{K, 0 \leq \epsilon_i \leq 1} \hspace{0.5cm} &\text{trace}(P) \nonumber\\
\text{where} \hspace{0.5cm}  (I-(K^{\epsilon})^T) P (I-K^{\epsilon}) &+ I = P
\end{align} 
Observe when $\epsilon_i = 1$ for all $i$ we have $K^{\epsilon} = K_{\text{sym}}$, \textit{i.e.}, the symmetric controller discussed before.  On the other extreme is when $\epsilon_i = 0$ for all $i$ for which $K^{\epsilon}$ is a purely asymmetric controller. Then for the optimal asymmetric controller problem we have that 
\begin{proposition}
	\label{optimal_asym_scaling}
	For the optimal asymmetric controller in Eq.~\eqref{opt_controller_asm} we have
	\begin{align*}
	\Hcal_2(I - K_{\text{asm}}) &= O(n^{3/2})
	\end{align*}
	\end{proposition}
Propositions~\ref{optimal_scaling},\ref{optimal_asym_scaling} show that the network energy of the optimal symmetric controller is higher than the asymmetric one. The underlying cause of this is distance to instability $(1 - \rho(I - K_{\theta}))$ where $\theta \in \{\text{asm}, \text{sym}\}$. To that end, we look at the scaled $\Hcal_2$--norm of the controller. For symmetric case, by using Proposition~\ref{norm_sing}, we have that
\begin{align}
\label{asymp_stable}
\Hcal_2^{S}(I - K_{\text{sym}}) = O(1)
\end{align}
On the other hand for $K_{\text{asm}}$ we have
\begin{equation}
\label{asm_asymp_stable}
\Hcal_2^{S}(I - K_{\text{asm}}) = \Theta(n^{3/2})
\end{equation}
Eq.~\eqref{asymp_stable} suggests that $I - K_{\text{sym}}$ is asymptotically unstable, \textit{i.e.},
$$\rho(I - K_{\text{sym}}) \rightarrow 1, \hspace{2mm} n \rightarrow \infty$$
This also means that if we suppress the effect of distance to instability then there is no growth in network energy (as its size increases). One can conclude that the key to making an undirected network more robust is by increasing the distance to instability. A similar argument does not extend to directed networks. Indeed as discussed before, even if distance to instability is large a network may exhibit arbitrary scaling. Then to make such a network more robust we use Theorem~\ref{symmetrizing}.

 We show how balancing affects network energy in Fig.~\ref{symmetrizing_random}. For this case we generate networks, $A_n$, such that $[A_n]_{ij} \sim \Ncal(0, 1)$ are i.i.d. We scale them so that every network has spectral radius $\gamma < 1$(See~\cite{tao2008random} for details). $\Gamma_n$, as in Definition~\ref{balance_networks}, is generated in a random fashion such that $\rho(\Gamma_n) \leq \gamma$. Then as we vary $\epsilon$ it is observed that $\Hcal_2$--norm reduces as $\epsilon$ increases. 
\begin{figure}[h]
	\centering
	\includegraphics[width=1.2\textwidth]{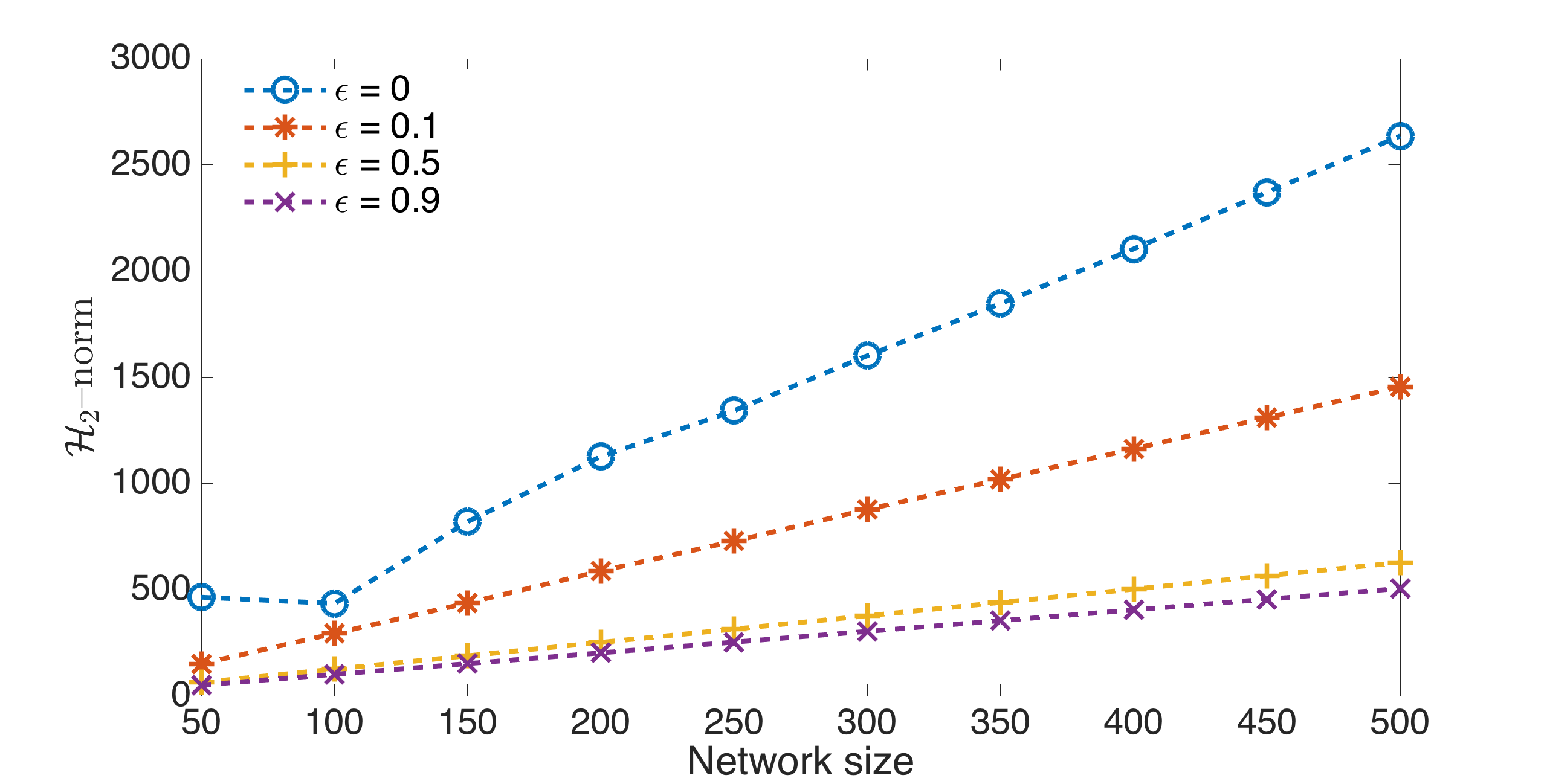}
	\caption{$\Hcal_2$--norm variation with different $\epsilon$--balancing}
	\label{symmetrizing_random}
\end{figure}
The effect of balancing the asymmetric controller is shown in Fig.~\ref{asm_controller}.
\subsection{Robustness as a Tail Risk Measure}
\label{non_normality}
{In many applications it is important to understand \textit{how far} network output deviates from its mean in response to a random shock. Such a ``worst'' case analysis is important in areas such as economics where one is interested in knowing how far the output of a factory, for example, may fall below optimum in response to an input shock. Although both tail risk and expected energy, $\Ecal(\cdot)$--norm, depend on certain network topological properties, the exact relation between them is unclear. In this section we make the connection that scaling in both expected energy and tail probability of output depend solely on the underlying network Gramian.}

In Eq.~\eqref{DT_LTI} assume that the shock $\omega \sim \mathcal{N}(\textbf{0}, I)$, \textit{i.e.}, $\omega_i \sim \Ncal(0, 1)$ and $\{\omega_i\}_{i=1}^n$ are independent. Define the aggregate output $x_{\infty}$ as
\begin{equation}
\label{aggregate_output}
x_{\infty} = \sum_{k=0}^{\infty}\textbf{1}^Tx(k)
\end{equation}
with $x(0) = \textbf{0}$ (for simplicity). In general, $x(0)$ is the predefined output, and the model here captures the deviation from this value in response to a shock. Then it is easy to check that $x_{\infty}$ is a linear combination of independent standard normal random variables and thus a Gaussian random variable(See~\cite{eisenberg2008sum}). We are interested in studying the distribution of aggregate output, $x_{\infty}$.
\\
The assumption that $\omega_i \sim \Ncal(0, 1)$ is merely illustrative. Our results hold for a general class of distributions that we define below:  
{
\begin{definition}
		\label{exp_tails}
	A random variable $X$ has exponential tails if 
		\begin{itemize}
		    \item $\Ebb\lbrack X \rbrack = \mu $
		    \item $\Ebb\lbrack (X - \mu)^2 \rbrack = \sigma^2$
		    \item $\Ebb \lbrack \exp{(bX)} \rbrack < \infty, \, \text{for some }b > 0$
		\end{itemize}
	We denote its distribution by $X \sim \Ecal(\mu, \sigma)$.
\end{definition}
For the remainder of this section, we will impose the following distributional assumption on $\omega$
\begin{assumption}
		\label{exp_tail}
		The shock, $\{\omega_i \sim \Ecal(0, 1)\}_{i=1}^n$, has continuous, symmetric probability density function with full support. Further, $\{\omega_i\}_{i=1}^n$ are independent.
\end{assumption}}
Example of $\Ecal(\mu, \sigma)$ is the standard normal where $\mu=0, \sigma=1$.
	
We focus on ``worst case'' analysis of aggregate output and will be interested in its tail probability behavior. From this perspective, $\Ecal(\mu, \sigma)$ is a richer class compared to the family of Gaussian distributions. At the same time, these distributions exhibit favorable tail behavior (such as the Large Deviation Principle~\cite{dembo_LD} etc.) that make them amenable to such analysis. We impose additional assumptions on the class of networks we consider in this section,	
\begin{assumption}
	\label{economic_assumption}
	Network matrix $A_n$ has the following properties
	\begin{itemize}
			\item $\lbrack A_n \rbrack_{ij} \geq 0$.
			\item There exists $v$ such that 
			\[
			A_nv = \lambda_{PF} v
			\]
			where $v_{\max}/v_{\min} = O(1)$ and $\lambda_{PF}$ is the Perron root.
	\end{itemize}
\end{assumption}
Examples of such networks include networks with degree normalization introduced in Section~\ref{sec:Preliminaries}. These fall under the more general class of networks where $A_n = \gamma W_n$ ($0 < \gamma < 1$) and $W_n$ is row--stochastic. Note that this assumption is weaker than those required for irreducibility or aperiodicity. For example the star network, $S_n$, is neither aperiodic nor irreducible, however $v_{\max} = v_{\min}$ (as $v = \textbf{1}$) in that case. 
\\
We can now define tail risk in a network as follows
\begin{definition}[Tail Risk]
	\label{tail_risk}
	For $x_{\infty}$ in response to $\{\omega_i \sim \Ecal(0, 1)\}_{i=1}^n$, define
	\begin{equation*}
	R_n(z) = -(1/n)\log{(\Pbb(|x_{\infty}| >  nz))}
	\end{equation*}
	Then $A_n$ exhibits \textbf{no tail risk} if there exists $z > 0$ such that $\lim_{n \rightarrow \infty} R_n(z) >0$.
\end{definition} 
$R_n(\cdot)$ is a measure of how likely it is for the average aggregate output to fall below a certain threshold $-z$ in response to a shock. In this context, a robust large network is resistant to random shocks, \textit{i.e.}, output is not affected too much.
{
\begin{example}
	Consider the special case when $A_n = \mu I_{n \times n}$. A simple computation shows us that 
	$$x_{\infty} = \frac{\sum_{i=1}^n \omega_i}{1 - \mu}$$
	Then by Cramer's theorem we have that $\lim_{n \rightarrow \infty}R_n(z) > 0$. 
	
	Next, when $A_n = \mu \lbrack 1, 1, \ldots, 1\rbrack^T \lbrack 1, 0, \ldots, 0\rbrack $, \textit{i.e.}, 
	$$x_i(k+1) = \mu x_1(k)$$ 
	then 
	$$x_{\infty} = \frac{(n-1)\mu + 1}{1-\mu}\omega_1 + \sum_{i=2}^n \omega_i$$ 
	and then it is easy to show that
	$$\lim_{n \rightarrow \infty}R_n(z) = 0$$ 
    This follows because 
    \begin{align*}
        &\Pbb(|x_{\infty}| > nz) \geq \Pbb\Big(\frac{(n-1)\mu + 1}{1-\mu}\omega_1 > nz, \sum_{i=2}^n \omega_i \geq 0\Big) \\
        &=\Pbb\Big(\frac{(n-1)\mu + 1}{1-\mu}\omega_1 > nz \Big)\Pbb\Big( \frac{\sum_{i=2}^n  \omega_i}{\sqrt{n-1}} \geq 0\Big) \\
        &\lim_{n \rightarrow \infty} \Pbb(|x_{\infty}| > nz) = \frac{1}{2}\Pbb\Big(\frac{\mu}{1-\mu} \omega_1 > z\Big) > 0
    \end{align*}
	In the first case, the shock is averaged out across the nodes: and the network exhibits no risk. In the second case output of each node, $x_i$, is heavily dependent on the other: and network exhibits tail risk. 
\end{example}}
{
In large networks, a fast decay of ``failure'' or ``shortfall'' probability with network dimension is ideal. $R_n(\cdot)$ captures this notion of risk in networks -- the more dependent individual components become, more likely the failure. Another interpretation that follows from Definition~\ref{tail_risk} is that whenever we have 
$$\frac{x_{\infty}}{\sqrt{n}} \sim \Ncal(0, \sigma^2_0)$$
for some $\sigma_0 = O(1)$, the network does not exhibit tail risk. We will show in Theorem~\ref{daron} that the converse is true as well. For example, the tail of $x_{\infty}/\sqrt{n}$ obtained from the star network $S_n$ becomes wider as $n$ increases and it will be shown in Proposition~\ref{macro} that star networks exhibit tail risk(See Fig.~\ref{tail_dist}). 
\begin{figure}[h]
	\centering
	\includegraphics[width=\columnwidth]{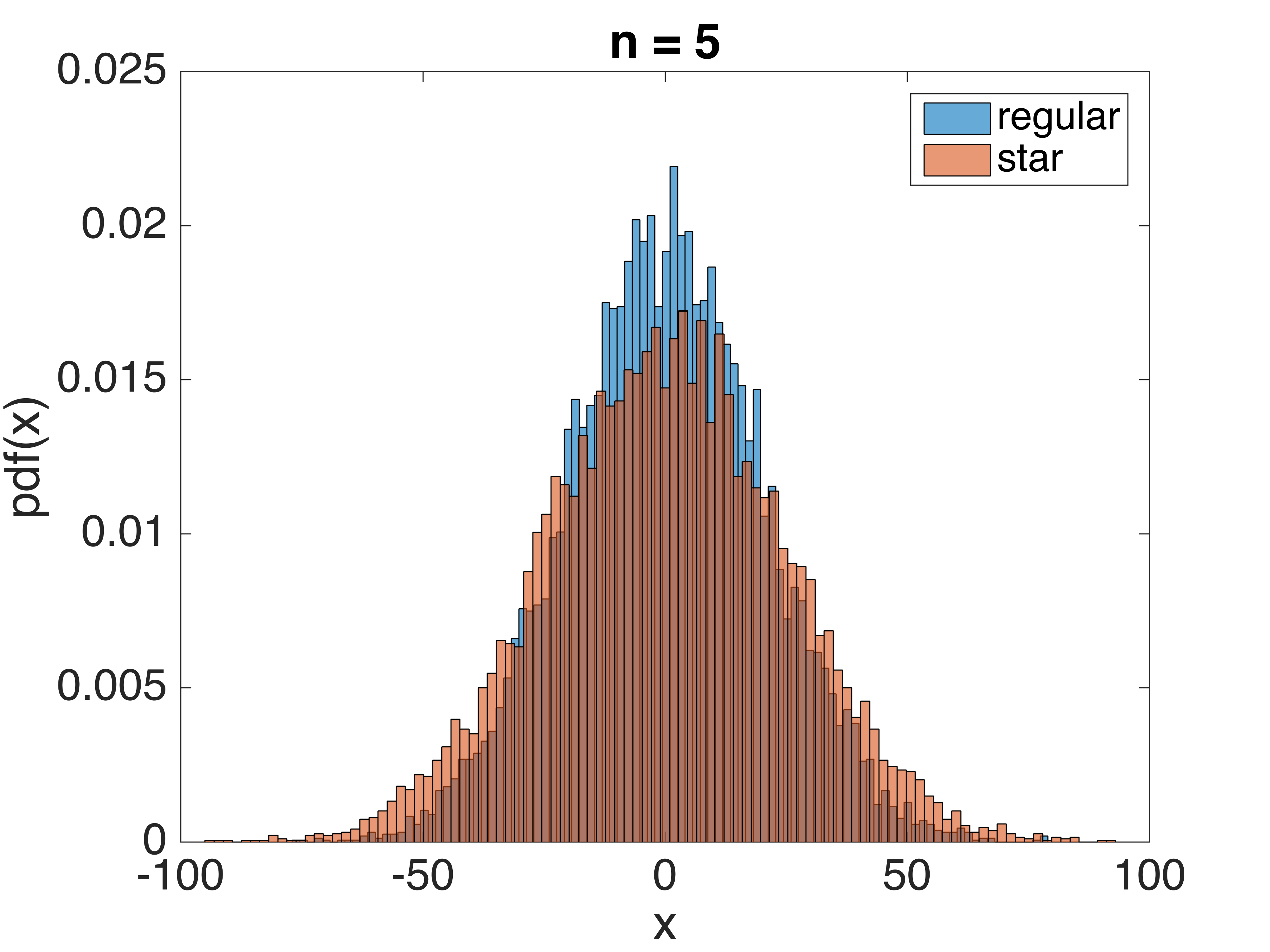}
	\includegraphics[width=\columnwidth]{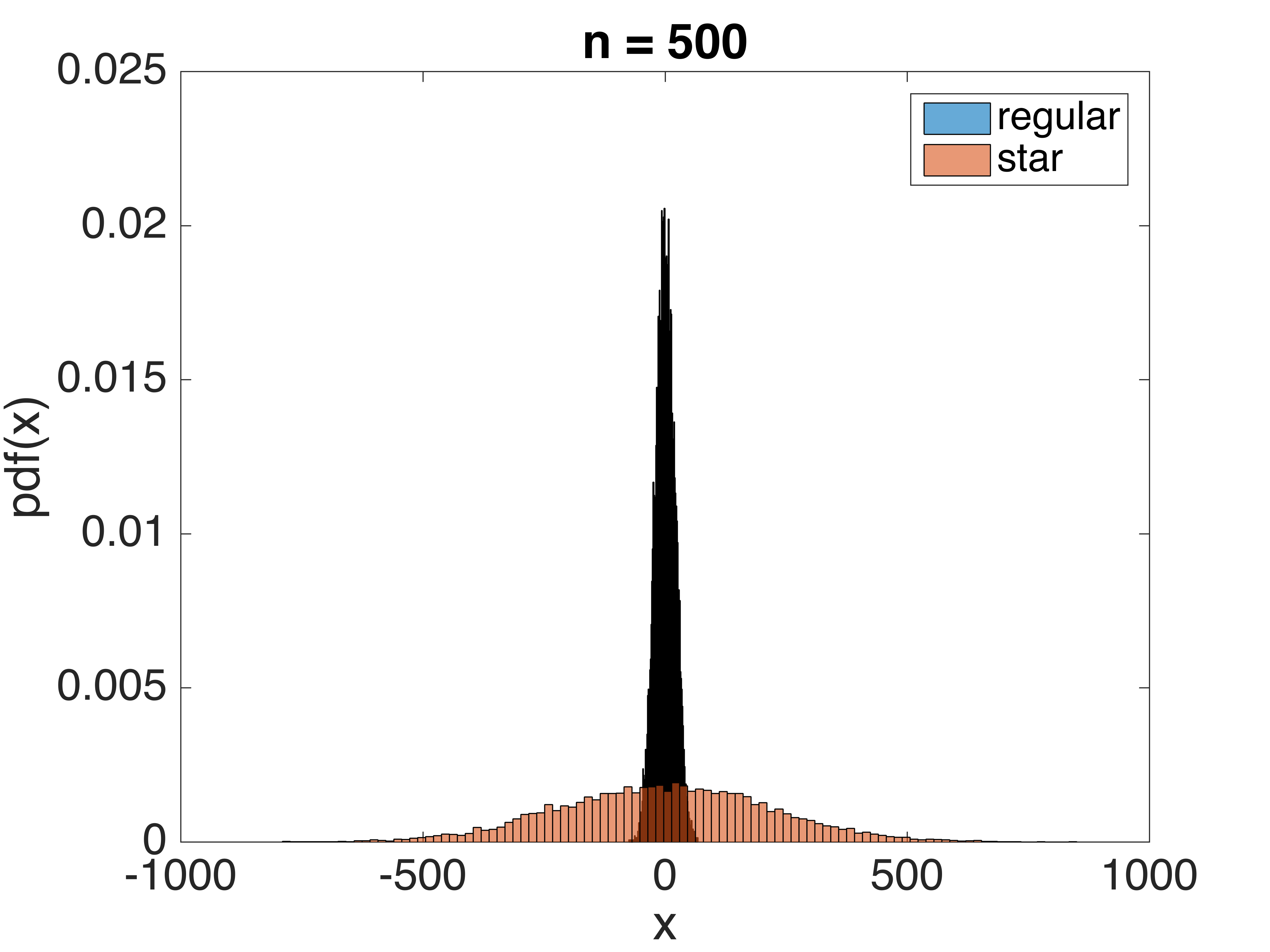}
	\caption{pdf of $x_{\infty}/\sqrt{n}$ as $n$ increases for regular and star networks}
	\label{tail_dist}
\end{figure}
The reason behind this is that $\sigma(x_{\infty}) \rightarrow \infty$ as $n \rightarrow \infty$ for $S_n$. This notion of comparison to Gaussian tail probability was introduced by Acemoglu et. al.(See~\cite{alireza_daron},\cite{daron_aggreg}) to measure disruption in economic production networks. To formalize this first we define ``tail risk relative to standard normal'' for random variables as this will give us some understanding of the behavior we intend to capture. 
\begin{definition}
	\label{z_tailrisk}
	Random variable $X$ exhibits \textbf{tail risk relative to standard normal distribution} if $\lim_{\tau \rightarrow \infty} r_X(\tau) = 0$ where 
		\begin{align*}
	r_X(\tau) = \dfrac{\log{(\Pbb(Z_X  < -\tau))}}{\log{(\Phi(-\tau))}}
	\end{align*}
	where 
	$$Z_X = \frac{X - \mu(X)}{\sigma(X)}$$
\end{definition}
Definition~\ref{z_tailrisk} implies that whenever $\tau$--tail probability (for very large $\tau$) of $X$ substantially exceeds that of Gaussian random variable, $X$ exhibits tail risk relative to standard normal. $r_X(\tau)$ can, in general, be hard to visualize. A thumb rule for any $X \sim \Ecal(0, 1)$ and large enough $\tau > 0$ is to inspect the tails of the pdf of $X$. This follows from: item 3 in Definition~\ref{exp_tails}, finiteness of moment generating function(See Chernoff's bound Theorem 2.29 Chapter 2~\cite{mdp2015}), and the approximation, for every $x > \tau$, that is shown below
\[
\Pbb(Z_X < -x) \approx e^{-\Omega(x)} \approx \int_{-\infty}^{-x} f_X(z) dz \approx f_X(-x)
\]
Here $f_X(z)$ is the pdf of $X$ at $z$. 
\begin{figure}[h]
	\centering
 	\includegraphics[width=1.1\textwidth]{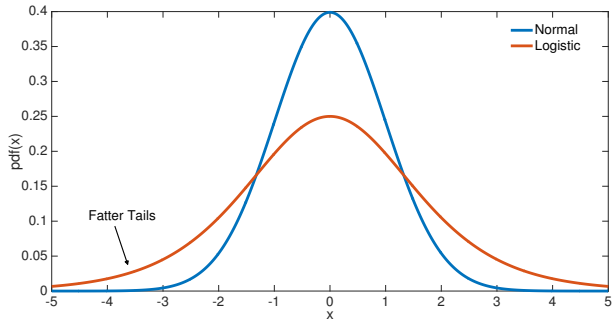}
	\caption{pdf of Logistic (red) vs Normal tails (blue); Logistic distribution has fatter tails and exhibits tail risk relative to standard normal}
	\label{tails}
\end{figure}
Figure~\ref{tails} shows a pair of distributions in the family $\Ecal(0, 1)$: Gaussian and logistic distribution. The pdf of logistic distribution, $f_L(x)$, is given by 
$$f_L(x) = \frac{e^{-x}}{(1+ e^{-x})^2}$$
Logistic distribution exhibits tail risk relative to standard normal distribution due to ``fatter tails''. Now, we extend the previous definition to the context of networks by setting $\tau = \sqrt{n}$. This captures size dependent scaling and follows the definition of macroeconomic tail risk in~\cite{alireza_daron},\cite{daron_aggreg}.
\begin{definition}[Macroeconomic Tail Risk]
	\label{macro_risk}
	The network, $A_n$, exhibits tail risk relative to the standard normal (or \textbf{macroeconomic tail risk}) if $\lim_{n \rightarrow \infty} R_{x_{\infty}}(\sqrt{n}) = 0$ where 
	\begin{align*}
	R_{x_{\infty}}(\tau) = \dfrac{\log{(\Pbb(Z_{x_{\infty}} < -\tau ))}}{\log{(\Phi(-\tau))}}
	\end{align*}
	where 
	$$ Z_{x_{\infty}} = \frac{x_{\infty} - \mu(x_{\infty})}{\sigma(x_{\infty})}$$
\end{definition}}
Tail risk differs from macroeconomic tail risk in the sense that macroeconomic tail risk (after normalization by its mean and variance) is measured relative to the standard normal distribution. Both notions of risk are applied extensively in economic networks and financial systems and as a result, it becomes important to understand how they relate to each other.
\begin{theorem}
	\label{daron_equivalence}
	Under Assumption~\ref{exp_tail},\ref{economic_assumption}, $A_n$ exhibits no tail risk if and only if 
	\[
		 ||P(A_n/\sqrt{\lambda_{PF}})||_1 = \Theta(1)
	\]
	where $\lambda_{PF}$ is its Perron root.
\end{theorem}
We showed that a lack of tail risk is equivalent to constant scaling in $\Lcal_1$--norm of an appropriate system Gramian. In the following discussion we show that tail risk coincides with macroeconomic tail risk and its application to a real network setting. \vspace{3mm}\\
\textbf{Tail Risk in Input Output Networks}\\
{
In this section we discuss applications of our results to the $2007$ U.S. Commodity market. The data for this has been taken from~\cite{alireza_daron}. It consists of a network with $379$ sectors where commodity generated by each sector is either utilized by itself or fed as input to another sector. These interconnections generate a network matrix, $A_{379}$. Some properties of this network are 
\begin{itemize}
	\item $[A_{379}]_{ij} \geq 0$ 
	\item $\rho(A_{379}) \simeq 0.51$
	\item $A_{379}$ (almost) satisfies Assumption~\ref{economic_assumption}.
\end{itemize}
Now it can be shown that 
\begin{equation}
\label{deficit_eq}
y_{t+1} =  A_{379} y_t + \omega_t
\end{equation}
Here $y_t$ is output deficit at time $t$ in response to an input shock $\omega_t$. Details of the economic model is given in Section II, III up to Eq. 25 in~\cite{long1983real}, Section I, VII in~\cite{alireza_daron}. This is summarized briefly in Section~\ref{economic_setup}.
\begin{figure}[h]
	\centering
	\includegraphics[width=1.2\textwidth]{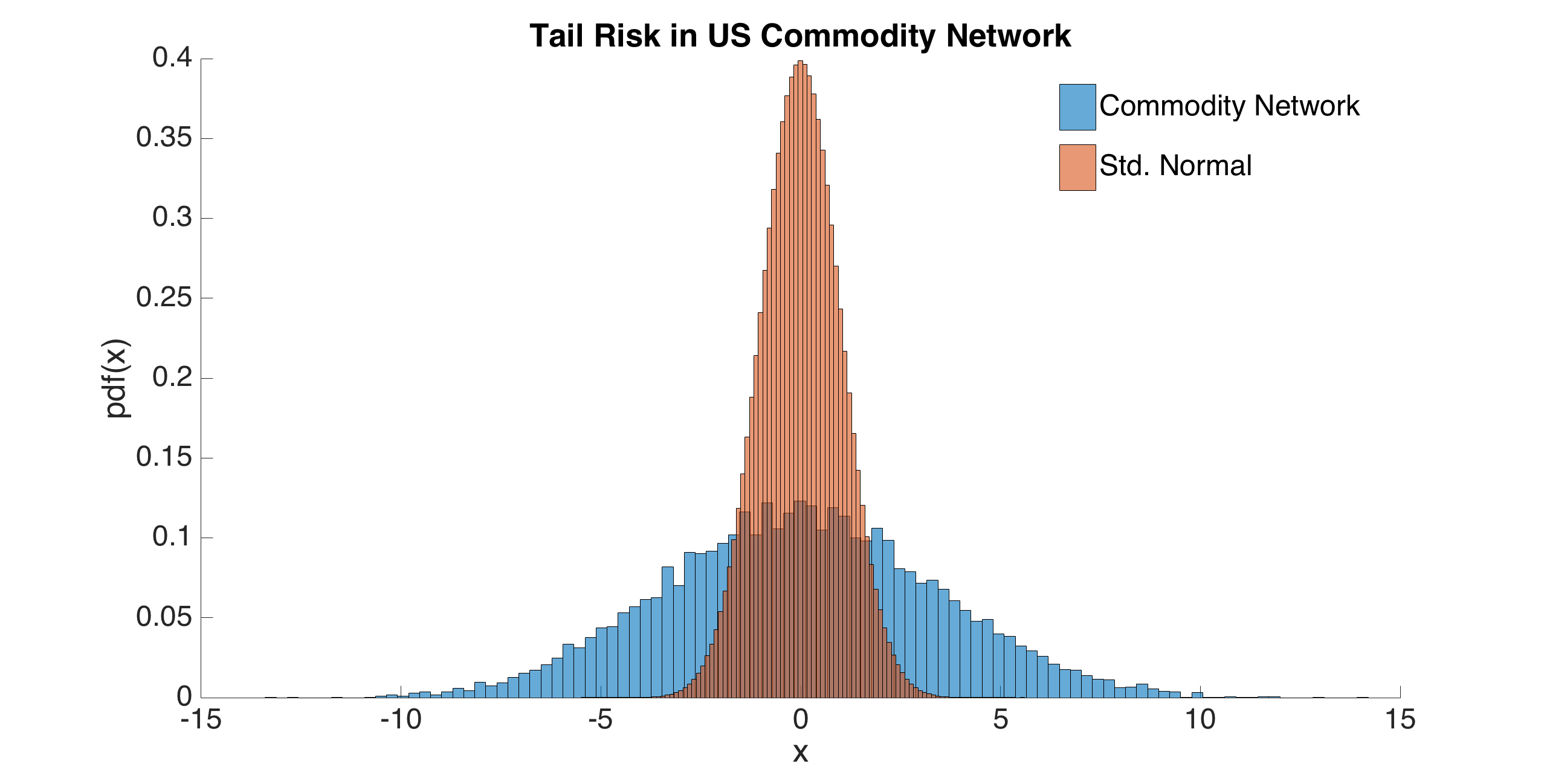}
	\caption{pdf of $x_{\infty}/\sqrt{n}$ for the US Commodity network vs. standard normal; US Commodity network exhibits tail risk as pdf tails of $x_{\infty}/\sqrt{n}$ are fatter than std. normal}
	\label{us_tails}
\end{figure}
Then $x_{\infty} = \sum_{t=0}^{\infty} \textbf{1}^T y_t $ is the aggregate deficit. The aggregate deficit is useful in estimating business cycles and computing gross domestic product, or GDP (See Proof of Proposition 1~\cite{alireza_daron}). We observe that the tails of aggregate output deficit are indeed fatter compared to the standard normal distribution. The question is how this observation can be related to Theorem~\ref{daron_equivalence}. In the next theorem we demonstrate that, under mild conditions,  a network exhibits tail risk whenever $x_{\infty}$ has fatter tails. 

\begin{theorem}
	\label{daron}
	Under Assumption~\ref{exp_tail},\ref{economic_assumption}, $A_n$ exhibits no macroeconomic tail risk if and only if it exhibits no tail risk.
\end{theorem}

Theorem~\ref{daron} suggests that the U.S. Commodity market network exhibits tail risk and this observation is corroborated through Fig. 1 of~\cite{alireza_daron}.} Finally, through Theorems~\ref{daron_equivalence},\ref{daron} we show that different notions of risk can be represented as robustness measures, \textit{i.e.}, manifestation of tail risk in a network is equivalent to poor scaling of some Gramian norm. We conclude with some examples of network topologies and comment on their robustness (to tail risk).

\begin{proposition}
	\label{macro}
The regular network, $R_n$, and cycle network, $C_n$, exhibits no macroeconomic tail risk. On the other hand, star network, $S_n$, exhibits macroeconomic tail risk.
\end{proposition}

It is perhaps no coincidence that the star network which has a high degree centrality (See~\cite{newman2008mathematics}) demonstrates tail risk. A generalization of degree centrality is eigenvector centrality.
\begin{definition}
	An irreducible network matrix, $A_n$, demonstrates no eigenvector centrality if 
	\[
	\pi_{\max}/\pi_{\min} = \Theta(1)
	\]
	where $\pi^T A_n = \pi^T$.
\end{definition}
{
\begin{corollary}
	\label{centrality}
	If an irreducible network matrix, $A_n$, has no eigenvector centrality then it exhibits no tail risk.
\end{corollary}
It is not hard to verify that $S_n$ has a central node and $R_n, C_n$ exhibit no centrality.} In addition to showing that tail risk is closely related to network energy, we also relaxed some of the assumption on network structure imposed by Acemoglu et. al. in~\cite{alireza_daron}. Our analysis can be extended to the case when we have increasing returns to scale, \textit{i.e.}, $\sum_{j=1}^na_{ij} > 1$ for example.

\section{Conclusions}
\label{sec:conclusion}
This work developed a general framework to study performance degradation in large networks. It is shown that dimension--dependent scaling of network energy is closely related to such degradation. The results here provide an analytical characterization of network energy scaling for an undirected network in terms of its distance to instability. It is shown that, for a fixed distance to instability, $\Hcal_2$--norm scales at most linearly in the size of network. On the contrary, results here suggest that for directed networks distance to instability and $\Hcal_2$--norm scaling may be completely unrelated. The observation that undirected networks exhibit a certain optimality in energy scaling could be used as a design principle to obtain more graceful performance. Under this general framework, notions of centrality and tail risk that appear more commonly in areas such as economics, finance etc. are related to network energy. This is a step towards more interdisciplinary research in design of robust network controllers. 

The work here is confined to networks with linear stable dynamics. Further extensions to marginally stable dynamics, such as those in consensus networks, will be discussed in a future work. 
%
%
\bibliographystyle{sources/IEEEtran}
\bibliography{sources/IEEEfull.bib}

\section{Appendix}
\subsection{Proof of Proposition~\ref{norm_sing2}}
Consider the case when $\omega$ is deterministic with $||\omega||_2 \leq 1$, then we 
\begin{align*}
\sum_{t=0}^{\infty}x(t)^T x(t) &= \omega^T \sum_{t=0}^{\infty}(A^{T})^t A^t\omega \\
\sup_{||\omega||_2 \leq 1}\sum_{t=0}^{\infty}x(t)^T x(t) &= \sup_{||\omega||_2 \leq 1} \omega^T \sum_{t=0}^{\infty}(A^{T})^t A^t\omega \\
&= \sup_{||\omega||_2 \leq 1} \omega^T P(A)\omega \\
\Mcal(A)&= \sigma_1(P)
\end{align*}
For the case when $\Ebb[\omega] = 0, \Ebb[\omega \omega^T] = I$ we have that 
\begin{align*}
\Ebb[\sum_{t=0}^{\infty}x(t)^T x(t)] &= \Ebb[\omega^T \sum_{t=0}^{\infty}(A^{T})^t A^t\omega] \\
&= \Ebb[ \text{trace}( \sum_{t=0}^{\infty} A^t \omega \omega^T (A^{T})^t) ] \\
&= \text{trace}(\sum_{t=0}^{\infty} A^t  (A^{T})^t)\\
\Ecal(A)&= (1/n)\text{trace}(P(A)) = (1/n) \Hcal_2(A)
\end{align*}

\subsection{Proof of Proposition~\ref{norm_sing}}
Since $A$ is symmetric, we have that $P(A) = \sum_{t=0} A^{2t}$. Now $||A^{2t}||_2 = \rho(A)^{2t}$ and since the eigen vectors of $A^{2t}$ identical to $A^{2t+2}$ we have
\[
\sigma_1(P(A)) = \frac{1}{1 -\rho^2(A)} \leq  \frac{1}{1 -\rho(A)}
\]
The assertion follows because $\Mcal(A) \geq \Ecal(A)$.

\subsection{Proof of Proposition~\ref{network_robustness}}
\label{robustness_proof}
For brevity, we prove here the scaling of the networks we introduced in Section~\ref{sec:Preliminaries}. \vspace{3mm}\\
\textbf{Star Network $(S_n)$}:\\
	\begin{align*}
	S_n=\gamma \begin{bmatrix}
	0  & \frac{1}{n-1} & \frac{1}{n-1} & \dots & \frac{1}{n-1} \\
	1       & 0 & 0 & \dots & 0 \\
	\vdots & \vdots & \vdots & \ddots & \vdots \\
	1       & 0 & 0 & \dots & 0
	\end{bmatrix}
	\end{align*}
	We can verify that 
	\begin{align*}
	S^{2k+1}_n=\gamma^{2k+1}\begin{bmatrix}
	0  & \frac{1}{n-1} & \frac{1}{n-1} & \dots & \frac{1}{n-1} \\
	1       & 0 & 0 & \dots & 0 \\
	\vdots & \vdots & \vdots & \ddots & \vdots \\
	1       & 0 & 0 & \dots & 0
	\end{bmatrix}
	\end{align*}
	and 	
	\begin{align*}
	S^{2k}_n=\gamma^{2k}\begin{bmatrix}
	1       & 0 & 0 & \dots & 0 \\
	0  & \frac{1}{n-1} & \frac{1}{n-1} & \dots & \frac{1}{n-1} \\
	\vdots & \vdots & \vdots & \ddots & \vdots \\
	0  & \frac{1}{n-1} & \frac{1}{n-1} & \dots & \frac{1}{n-1} 
	\end{bmatrix}
	\end{align*}
Since $P(S_n) \succeq_{\text{PSD}} \sum_{k=0}^{\infty} (S_n^T)^{2k+1}S_n^{2k+1}$ we have that $||P(S_n)||_2 \geq \gamma(1-\gamma^2)^{-1} ||S_n||^2_2$. But 
\[
P(S_n) = I + \sum_{k=0}^{\infty} \gamma^{2k}(S_n^2)^TS_n^2 + \sum_{j=0}^{\infty} \gamma^{2j+1}(S_n)^TS_n
\]
thus, $||P(S_n)||_2 \leq 1 + (1-\gamma^2)^{-1}(||S_n||^2_2 + ||S_n^2||^2_2)$ which follows from the subadditivity of $\Lcal_2$--norm. Further, since 
\[
||S_n^T S_n||_{\text{max}} \leq ||S_n^T S_n||_{2} \leq ||S_n^T S_n||_{S}
\]
we have that $||S_n||^2_2 = \Theta(n)$ and since $S_n^2$ is symmetric we have $||S_n^2||^2_2 = \Theta(1)$, then our claim follows. By a similar argument, $||P(S_n)||_S = \Theta(n)$. \vspace{3mm}\\
\textbf{Regular Network $(R_n)$}:	
	\begin{align*}
	R_n=\gamma \begin{bmatrix}
	0  & \frac{1}{n-1} &\dots & \frac{1}{n-1} \\
	\frac{1}{n-1}       & 0 &  \dots & \frac{1}{n-1} \\
	\vdots & \vdots &  \ddots & \vdots \\
	\frac{1}{n-1}       & \frac{1}{n-1} & \dots & 0
	\end{bmatrix}
	\end{align*}
Since $R_n$ is symmetric then $||P(R_n)||_2 = (1/(1 - \rho^2(R_n)))^{-1}$ and $\rho(R_n) = \gamma$ (By Perron-Frobenius Theorem) and from Proposition~\ref{norm_sing} we have $||P(R_n)||_2 = \Theta(1)$ and due to the symmetry of topology we have $\text{trace}(P(R_n)) = \Theta(n)$. \vspace{3mm}\\
\textbf{Directed Line $(DL_n)$}:
\begin{align*}
DL_n=\begin{bmatrix}
0  & 0  & \dots & 0 \\
1  & 0 & \dots & 0 \\
\vdots & \ddots  & \ddots & \vdots \\
0    & \dots & 1 & 0
\end{bmatrix}
\end{align*}
$DL_n$ is nilpotent, \textit{i.e.}, $DL^{n}_n = 0_{n \times n}$. One can check that $(DL^k_n)^TDL^k_n = \text{Diag}(\underbrace{1, 1, \ldots}_\text{n-k}, \underbrace{0, 0, \ldots}_\text{k})$, where $\text{Diag}(v)$ is the matrix with diagonal elements as $v$. Then we have that $||P(DL_n)||_2 = ||\sum_{k=0}^{\infty}(DL^k_n)^TDL^k_n||_2 = ||\Delta_n||_2 = \Theta(n), ||P(DL_n)||_S = \Theta(n^2)$ where 
\begin{align*}
\Delta_n=\begin{bmatrix}
n  & 0  & \dots & 0 \\
0  & n-1 & \dots & 0 \\
\vdots & \ddots  & \ddots & \vdots \\
0    & \dots & 0 & 1
\end{bmatrix}
\end{align*}
\vspace{3mm} \\
\textbf{Cycle Network $(C_n)$}:
The proof is similar to that of $R_n$ (a combination of Proposition~\ref{norm_sing} and symmetric topology).  
\vspace{3mm} \\
\textbf{Wigner Ensemble}: 
The proof of this follows from Theorem A in~\cite{bai1988necessary}. There it is shown that for Wigner matrices (actually a much larger class of symmetric matrices) we have that $\rho(A) = 2 \sigma < 1$ almost surely as $n \rightarrow \infty$ \textit{i.e.} large networks. As a result in our case, using Proposition~\ref{norm_sing}, we have that 
\[
\Mcal(A), \Ecal(A) = O\Big(\frac{1}{1 - 2\sigma}\Big)
\] 
This result is again corroborated by Theorem 1 in~\cite{preciado2016controllability} (but convergence is shown in expectation only). From Eq. 10 in the same reference and picking $k=1$, we see that 
\begin{align*}
\frac{1}{n} \text{trace}\Big(\Ebb\{P\} \Big)&= \sum_{j=0}^{\infty} \sigma^{2j} \frac{1}{j+1} {2j \choose j} \\
\Ecal(A) &< 1  + \sum_{j=1}^{\infty} (2\sigma)^{2j}\frac{\sqrt{2}e}{2 \pi \sqrt{j}(j+1)} \\
&< \frac{C}{1 - 2 \sigma}
\end{align*}
for some absolute constant $C$. We also verify this numerically
\begin{figure}[h]
	\centering
	\includegraphics[width=\columnwidth]{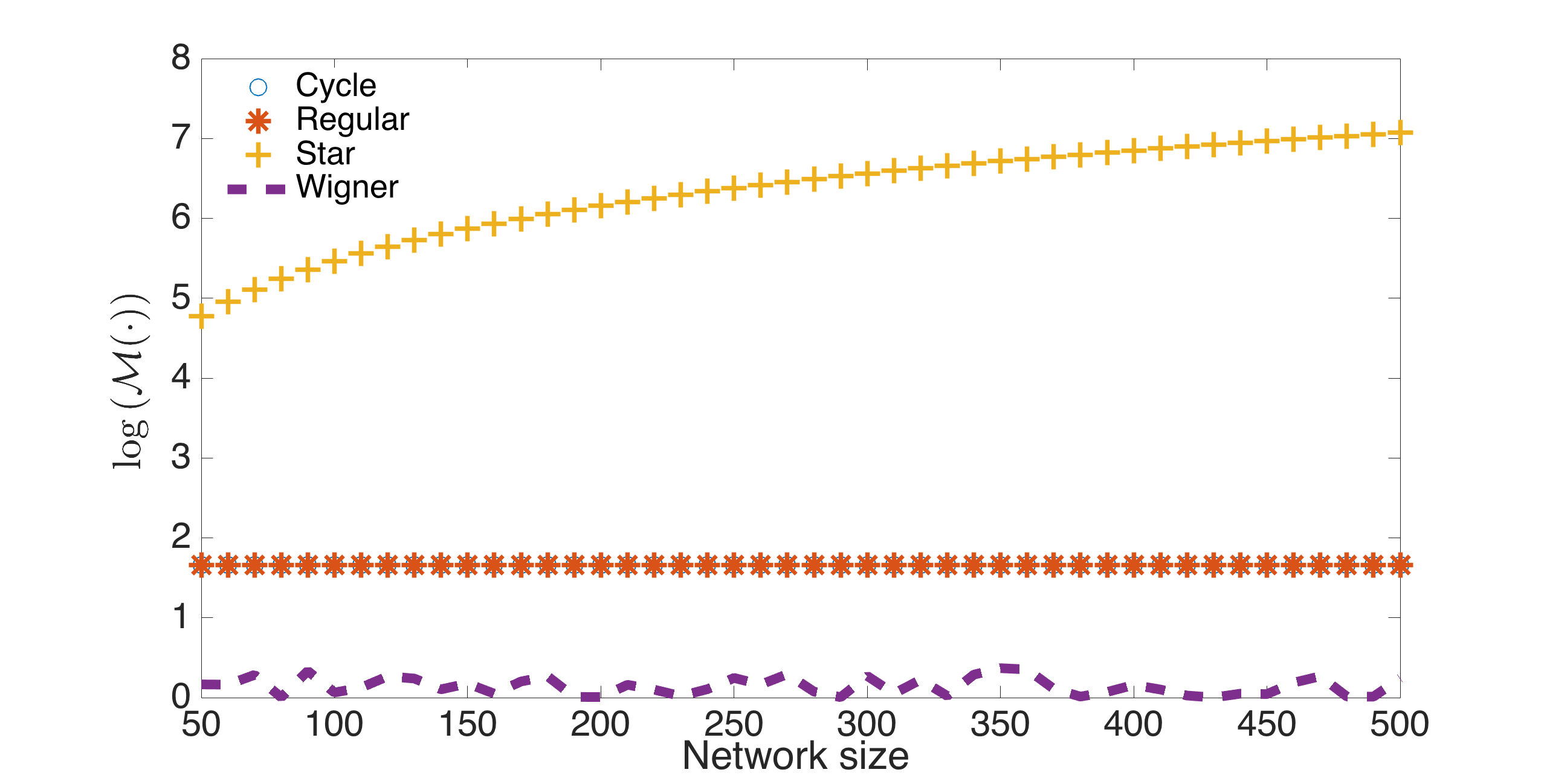}
	\caption{$\Mcal(\cdot)$ growth of networks with size}
	\label{network_growth}
\end{figure}
\subsection{Proof of Proposition~\ref{platoon_growth}}
\label{platoon_growth_proof}
\begin{align*}
	J=\begin{bmatrix}
		\lambda_1  &1 & 0 & \dots & 0 \\
		0  & \lambda_2 & 1 & \dots & 0 \\
		\vdots & \vdots & \vdots & \ddots & \vdots \\
		0 & 0 & 0 & \dots & \lambda_n
	\end{bmatrix}
\end{align*}
Notice that since $\lambda_i > 0$, we have that $||J^{d}||_{\infty} \geq ||J(\lambda)^d||_{\infty}$ for any $0 \leq \lambda \leq \min_i \lambda_i$, where 
\begin{align*}
J(\lambda)=\begin{bmatrix}
\lambda  &1 & 0 & \dots & 0 \\
0  & \lambda & 1 & \dots & 0 \\
\vdots & \vdots & \vdots & \ddots & \vdots \\
0 & 0 & 0 & \dots & \lambda
\end{bmatrix}
\end{align*}
Then $J^n(\lambda) = (\lambda I + N)^n = \sum_{r=0}^{n} {n \choose r}\lambda^r N^{n-r}$ where $N$ is the nilpotent matrix --
\begin{align*}
	N=\begin{bmatrix}
		0  &1 & 0 & \dots & 0 \\
		0  & 0 & 1 & \dots & 0 \\
		\vdots & \vdots & \vdots & \ddots & \vdots \\
		0 & 0 & 0 & \dots & 0
	\end{bmatrix}
\end{align*}
Then we have 
\begin{align*}
	J^n(\lambda)=\begin{bmatrix}
		\lambda^n  &{n \choose 1} \lambda^{n-1} & {n \choose 2} \lambda^{n-2} & \dots &  {n \choose 1} \lambda \\
		0  & \lambda^n & {n \choose 1} \lambda^{n-1} & \dots &  {n \choose 2} \lambda^2\\
		\vdots & \vdots & \vdots & \ddots & \vdots \\
		0 & 0 & 0 & \dots & \lambda^n
	\end{bmatrix}
\end{align*}
Thus $||J^n(\lambda)||_{\infty} = (1 + \lambda)^n - 1 \implies ||J^n||_{\infty} \geq (1 + \lambda)^n - 1$. Then $||J^n||_2 = \Omega(\exp{(cn)})$ for some $c > 0$ because  $||\cdot||_{\infty} \leq \sqrt{n}||\cdot||_2$. Since,
\begin{align*}
	P(J) &- (J^n)^TJ^n \succeq_{\text{PSD}} 0 \\
	&\implies ||P(J^n)||_2 \geq ||J^n||^2_2
\end{align*}
our claim is now proven. 
{
\subsection{Proof of Theorem~\ref{symmetrizing}}
For the proof here, fix matrix $A$, its Gramian $P(A)$. Let the SVD of $P(A) = U_n D_n U_n^T$, then define as 
\begin{equation}
\label{symmetrizer}
\Lambda = U_n \Gamma_n U_n^T
\end{equation}
where $\rho(\Gamma_n) \leq \rho(A) = \gamma$. Set $A^{\epsilon}$ as the $\epsilon$--balanced version of $A$.  We have a very simple observation that we state here for completeness.
\begin{lemma}
	\label{ones_psd}
$$\Lambda P(A) \Lambda \preceq_{PSD} \gamma^{2} P(A)$$
\end{lemma}
\begin{proof}
Clearly,
\[
\Lambda P(A) \Lambda = U_n D_n\Gamma^2_n U_n^T
\]
Then our assertion follows because $\Gamma_n \preceq \gamma I$.
\end{proof}

We will need the following lemmas for the proof of this theorem.
\begin{lemma}
	\label{trace_bnd}
For any matrices $A, B$ we have that 
\[
A^TB + B^TA \preceq_{PSD} A^TA + B^TB 
\]
\end{lemma}
\begin{proof}
This follows from observing that 
\[
(A-B)^T(A-B) \succeq_{PSD} 0
 \]
\end{proof}
\begin{lemma}
	\label{sum_trace_bnd}
	For any matrices $\{A_i\}_{i=1}^P$ we have that 
	\[
	\Big(\sum_{i=1}^P A_i^T \Big) \Big(\sum_{i=1}^P A_i \Big) \preceq_{PSD} P \Big(\sum_{i=1}^P A_i^TA_i \Big) 
	\]
\end{lemma}
\begin{proof}
By expansion we have
\begin{align*}
\Big(\sum_{i=1}^P A_i^T \Big) \Big(\sum_{i=1}^P A_i \Big) &= \sum_{i=1}^P A_i^TA_i \\
&+ \sum_{i=1}^P \sum_{j \neq i} A_i^T A_j \\
&= \sum_{i=1}^P A_i^TA_i + \sum_{1 \leq i < j \leq P} (A_i^T A_j + A_j^T A_i) \\
&\overset{(a)}{\preceq P} \Big(\sum_{i=1}^P A_i^TA_i \Big) 
\end{align*}
Here $(a)$ follows from Lemma~\ref{trace_bnd}.
\end{proof}

\begin{lemma}
\label{psd_lemma}
For $A$, the corresponding Gramian, $P(A)$ and $\Lambda$, then for any psd matrix $B$ of the following form
\[
B =\Lambda^{t^{\lambda}_k} \ldots (A^{T})^{t_2}\Lambda^{t^{\lambda}_1}(A^{T})^{t_{1}}A^{t_{1}}\Lambda^{t^{\lambda}_{1}}A^{t_2} \ldots \Lambda^{t^{\lambda}_k}
\]
where $\sum_{j=1}^k t^{\lambda}_{j} = t^{\lambda}$ then 
	$$B \preceq \gamma^{2t^{\lambda}} P(A)$$
\end{lemma}
\begin{proof}
For any $t_1, t_2$ we have that
\[
\Lambda^{t^{\lambda}_1} P(A) \Lambda^{t^{\lambda}_1} \succeq_{PSD} \Lambda^{t^{\lambda}_1} (A^{t_2})^TA^{t_2} \Lambda^{t^{\lambda}_1}
\]
Next observe that from Lemma~\ref{ones_psd}
\[
\gamma^{2t_1}P(A) \succeq_{PSD} \Lambda^{t_1} P(A) \Lambda^{t_1}
\]
Then we have 
\begin{align*}
B &=\Lambda^{t^{\lambda}_k} \ldots (A^{T})^{t_2}\Lambda^{t^{\lambda}_1}(A^{T})^{t_{1}}A^{t_{1}}\Lambda^{t^{\lambda}_{1}}A^{t_2} \ldots \Lambda^{t^{\lambda}_k}\\
&\preceq_{PSD} \Lambda^{t^{\lambda}_k} \ldots \gamma^{2t^{\lambda}_1} (A^{T})^{t_2} P(A) A^{t_2} \ldots  \Lambda^{t^{\lambda}_k}
\end{align*}
and the proof follows by recursively applying the inequalities.
\end{proof}
A direct conclusion of Lemma~\ref{psd_lemma} is that
\begin{align*}
\sigma_1(B) &\leq \gamma^{2t^{\lambda}} {\sigma_1(P(A))} \\
\text{tr}(B ) &\leq \gamma^{2t^{\lambda}} \text{tr}(P(A))
\end{align*} 
we will use these observations in our proof.
\\
In the proof we will show first how $\sigma_1(P(A^{\epsilon}))$ relates to $\sigma_1(P(A))$. For any $A^{\epsilon} = (1 - \epsilon) A + \epsilon \Lambda$ where $\Lambda$ is as in Eq~\eqref{symmetrizer}. Observe that, where $\Aeps = (1 - \epsilon)A, \Leps = \epsilon \Lambda$ for shorthand, 
\begin{align}
( \Aeps + \Leps)^k = \Aeps^k &+ \underbrace{\Aeps^{k-1} \Leps + \Aeps^{k-2} \Leps \Aeps + \ldots}_{{k \choose 1} \text{arrangements}}  \nonumber \\
&+ \underbrace{\Aeps^{k-2} \Leps^2 + \Aeps^{k-3} \Leps^2 \Aeps + \ldots}_{{k \choose 2} \text{arrangements}} \label{arrangements} \\
||( \Aeps + \Leps)^k|| \leq ||\Aeps^k|| &+ \underbrace{||\Aeps^{k-1} \Leps|| + ||\Aeps^{k-2} \Leps \Aeps|| + \ldots}_{{k \choose 1} \text{arrangements}}  \nonumber \\
&+ \underbrace{||\Aeps^{k-2} \Leps^2|| + ||\Aeps^{k-3} \Leps^2 \Aeps|| + \ldots}_{{k \choose 2} \text{arrangements}} \label{arrangements_norm} \\
&\leq (1-\epsilon)^k \sqrt{\sigma_1(P(A))} \nonumber \\
&+ {k \choose 1} (1 -\epsilon)^{k-1} \gamma \epsilon \sqrt{\sigma_1(P(A))} \nonumber \\
&+ {k \choose 2} (1 -\epsilon)^{k-2} (\gamma \epsilon)^2 \sqrt{\sigma_1(P(A))} + \ldots  \nonumber
\end{align}
The last inequality follows from Lemma~\ref{psd_lemma}. Then, whenever $\rho(\Lambda) \leq \gamma$ 
\begin{align}
\sigma((A_{\epsilon} + \Lambda_{\epsilon})^k) &\leq \sum_{j=0}^k {k \choose j} (1-\epsilon)^{k-j}\epsilon^{j} \sqrt{\sigma_1(P(A))}\gamma^j  \nonumber\\
&\leq (1 - \epsilon + \epsilon \gamma)^k \sqrt{\sigma_1(P(A))} \label{bnd_balance_pow}
\end{align}
Eq.~\eqref{bnd_balance_pow} directly implies that $A_{\epsilon} + \Lambda_{\epsilon}$ is stable (shown in Lemma~\ref{eps_balance}). Then $P(A^{\epsilon})$ exists and is uniquely given by 
$$P(A^{\epsilon}) = \sum_{k=0}^{\infty}((A^{\epsilon})^T)^{k}(A^{\epsilon})^{k}$$ 
It follows that 
\begin{align*}
\sigma_1(P(A^{\epsilon})) &\leq \sum_{k=0}^{\infty}(1 - \epsilon + \epsilon \gamma)^{2k}\sigma_1(P) \\
&=(1/(1 - (1 - \epsilon + \epsilon \gamma)^2)) \sigma_1(P) \\
&= O(\sigma_1(P(A)))
\end{align*}
\\
Next, we bound $\text{tr}(P(A^{\epsilon}))$ in terms of $\text{tr}(P(A))$. The key idea is to bunch together the terms that have same number of $\Lambda_{\epsilon}$. Consider the following example,
\begin{align*}
&(( \Aeps + \Leps)^2)^T( \Aeps + \Leps)^2 \\
&= (( \Aeps + \Leps)^2)^T(\Aeps^2 + \Aeps \Leps + \Leps \Aeps + \Leps^2) \\
&= (\Aeps^2)^T \Aeps^2 + (\Aeps^2)^T (\Aeps \Leps + \Leps \Aeps) \\
&+ (\Aeps^2)^T \Leps^2 \\
&+ (\Aeps \Leps + \Leps \Aeps)^T \Aeps^2 \\
&+ (\Aeps \Leps + \Leps \Aeps)^T (\Aeps \Leps + \Leps \Aeps)\\
& + (\Aeps \Leps + \Leps \Aeps)^T \Leps^2 \\
&+ \Leps^2 \Aeps^2 + \Leps^2 (\Aeps \Leps + \Leps \Aeps) + \Leps^4 \\
&\overset{(a)}{\preceq}_{PSD} \underbrace{3}_{2+1}((\Aeps^2)^T \Aeps^2  \\
&+ \underbrace{(\Aeps \Leps + \Leps \Aeps)^T}_{{2 \choose 1} \text{arrangements}} (\Aeps \Leps + \Leps \Aeps) \\
&+ \Leps^4)
\end{align*}

where $(a)$ follows from Lemma~\ref{trace_bnd}. Generalizing this observation we get 
\begin{align}
&(( \Aeps + \Leps)^k)^T( \Aeps + \Leps)^k \preceq_{PSD} (k+1)((\Aeps^T)^k\Aeps^k \nonumber\\
&+(\Aeps^{k-1} \Leps  + \ldots)^T \underbrace{(\Aeps^{k-1} \Leps + \Aeps^{k-2} \Leps \Aeps+ \ldots)}_{{k \choose 1} \text{arrangements}} \nonumber \\
&+ (\Aeps^{k-2} \Leps^2 + \ldots)^T\underbrace{(\Aeps^{k-2} \Leps^2 + \Aeps^{k-3} \Leps^2 \Aeps + \ldots)}_{{k \choose 2} \text{arrangements}} \nonumber\\
& + \ldots)\label{expand_sum}
\end{align}
\\
To each of the individual summands above we will apply Lemma~\ref{sum_trace_bnd}. 
\begin{align*}
&(\Aeps^{k-1} \Leps  + \ldots)^T \underbrace{(\Aeps^{k-1} \Leps + \Aeps^{k-2} \Leps \Aeps+ \ldots)}_{{k \choose 1} \text{arrangements}} \\ &\preceq_{PSD} {k \choose 1}\Big(\sum_{{k \choose 1} \text{arrangements}}\Leps(\Aeps^{k-1})^T\Aeps^{k-1} \Leps  + \ldots \Big) \\
&\preceq_{PSD} {k \choose 1}^2 (1-\epsilon)^{2k-2}(\gamma \epsilon)^2 P(A)
\end{align*}
\\
Here the last inequality follows again from Lemma~\ref{psd_lemma}. Using a similar upper bound for each of the summand in Eq.~\eqref{expand_sum} we get  
\begin{align*}
&\text{trace}((( \Aeps + \Leps)^k)^T( \Aeps + \Leps)^k)\\
& \leq (k+1)\Bigg(\sum_{j=0}^k {k \choose j}^2 (1-\epsilon)^{2(k-j)}(\gamma \epsilon)^{2j}\Bigg) \text{trace}(P)\\
& \leq (k+1)\Bigg(\sum_{j=0}^k {2k \choose 2j} (1-\epsilon)^{2(k-j)}(\gamma \epsilon)^{2j}\Bigg) \text{trace}(P)\\
&\leq (k+1) (1 - \epsilon + \epsilon \gamma)^{2k} \text{trace}(P)\\
&\text{trace}(P(\Aeps)) \leq \sum_{k=0}^{\infty}k (1 - \epsilon + \epsilon \gamma)^{2k} \text{trace}(P) \\
&\text{trace}(P(\Aeps)) = O(\text{trace}(P(A)))
\end{align*}
The proof relies on 
\[
{k \choose j}^2 \leq {2k \choose 2j}
\]
This is a simple combinatorial identity which has been proved in Proposition~\ref{combinatorial}.

Finally the stability of the symmetrized network is proved by Lemma~\ref{eps_balance}.
\begin{lemma}
	\label{eps_balance}
	For any $0 \leq \epsilon \leq 1$, and $\rho(A_n) = \gamma < 1$, we have that
	\[
	\rho(A^{\epsilon}_n) < 1
	\]
\end{lemma}
\begin{proof}
	For $\epsilon > 0$, we prove this by showing that the Gramian series exists, \textit{i.e.}, Theorem~\ref{symmetrizing}. This follows from Eq.~\eqref{bnd_balance_pow}
	\[
	\sigma(((1 - \epsilon)A+ \epsilon \Lambda)^k) \leq (1 - \epsilon + \epsilon \gamma)^k \sqrt{\sigma_1(P(A))}
	\]
	Then we know that for any $B$
	\[
	\lim_{k \rightarrow \infty} ||B^k||^{1/k} = \rho(B)
	\]
	Then we get that $\rho((1 - \epsilon)A+ \epsilon \Lambda) = (1 - \epsilon + \epsilon \gamma) < 1$ for $\gamma < 1$. When $\epsilon =0$, $\Aeps = A$ which is stable by assumption.
	
\end{proof}
}
\subsection{Proof of Proposition~\ref{optimal_scaling}}
For a symmetric stable matrix, the Gramian $P(A) = (I - A^2)^{-1}$. Since $(I - K)$ is symmetric then we have that for any $P = (I - (I-K)^2)^{-1}$.
\begin{align*}
P&=  (I - (I-K)^2)^{-1}\\
&= (2K - K^2)^{-1} \\
&= (1/2)K^{-1}(I - (1/2)K)^{-1}
\end{align*}
Since $K^2$ is positive semidefinite, and since $(I - K)$ is stable we have that $\rho(K) < 2$ since $\rho(I - K) =  |1 - \rho(K)|$. Then we have that $\rho((1/2)K) < 1$ and that all eigenvalues of $I - (1/2)K$ are nonnegative and it is positive semidefinite. 

Now we expand $(1/2)K^{-1}(I - (1/2)K)^{-1}$
\begin{align*}
(1/2)K^{-1}(I - (1/2)K)^{-1} &= (1/2)K^{-1} \\
&+ (1/4)I + (1/8)K + \\
&(1/16)K^2 (I -  (1/2)K)^{-1}
\end{align*}
We observe that $K^2 (I -  (1/2)K)^{-1}$ is psd because if $AB = (AB)^T$ for two psd matrices $A, B$ then $AB$ is psd. Then 
\begin{align*}
\text{trace}((1/2)K^{-1}(I - (1/2)K)^{-1}) &\geq n/4 +  \text{trace}((1/2)K^{-1} \\
&+ (1/8)K)
\end{align*}
Now our claim follows from Eq. (8), (9), (10) in~\cite{lin2012optimal}. Unlike there, here we have the following constraint for stability that $0 \leq k_i \leq 2$. The second claim follows because if $\Hcal_2$--norm is $\Omega(n^2)$ then the largest singular value has to be $\Omega(n)$.

\subsection{Proof of Proposition~\ref{optimal_asym_scaling}}
Consider the following controller 
\begin{equation*}
K_{\frac{1}{2}} = -(1/2)DL_n + (1/2)I_{n \times n}
\end{equation*}
$K_{\frac{1}{2}}$ is an asymmetric controller with spectral radius$=1/2$. Then $(I-K) = (1/2)DL_n + (1/2)I_{n \times n} = (1/2)J(1)$ (from Proof of~\ref{platoon_growth_proof}). Now for $t \leq n$,
\begin{align*}
\text{trace}((J^t(1))^TJ^t(1)) &\leq \sum^{t}_{r = 0}(t-r){t \choose r}^2 \\
&\leq t \sum^t_{r = 0}{t \choose r}^2\\
&= t {2t \choose t}\\
&= C \sqrt{t} 2^{2t}
\end{align*}
The last equality follows from Stirling's approximation. Therefore for $t \leq n$ we have, 
\begin{align*}
\sum_{k=0}^{n}\text{trace}((1/2^{2t})(J^k(1))^TJ^k(1)) &\leq C n \sqrt{n}
\end{align*}
For the other case when $t > n$, we have 
\begin{align*}
\text{trace}((J^t(1))^TJ^t(1)) &\leq \sum_{k=1}^{n}\sum_{r=1}^{n-k}{t \choose r}^2\\
&\overset{a}{=}\sum^{n-1}_{r = 0}(n-r){t \choose r}^2 \\
&\leq n \sum^{n-1}_{r = 0}{t \choose r}^2
\end{align*}
Equality $a$ comes from applying Fubini's theorem. We break our analysis into $t \in (n, 2n), (2n, 4n), \ldots$, \textit{i.e.},
\begin{align*}
\sum_{t=n+1}^{\infty}\text{trace}((1/2^{2t})(J^t(1))^TJ^t(1)) &\leq \sum_{t=n+1}^{\infty} \dfrac{n}{2^{2t}}\sum^{n-1}_{r = 0}{t \choose r}^2\\
&= \sum_{t=n+1}^{2n} \dfrac{n}{2^{2t}}\sum^{n-1}_{r = 0}{t \choose r}^2\\
&+\sum_{t=2n+1}^{4n} \dfrac{n}{2^{2t}}\sum^{n-1}_{r = 0}{t \choose r}^2 \ldots
\end{align*}
First for $t \in (n, 2n)$ we have the following upper bound
\begin{align*}
 \sum_{t=n+1}^{2n} \dfrac{n}{2^{2t}}\sum^{n-1}_{r = 0}{t \choose r}^2  &\overset{a}{\leq} C\sum_{t=n+1}^{2n}\dfrac{n}{2^{2t}}\dfrac{2^{2t}}{\sqrt{t}} \\
 &\leq C n \sqrt{n}
\end{align*}
Inequality $a$ comes from Stirling's approximation assuming $n-1=t$. The key idea of the proof is to divide the range $(0, t)$ into $(0, n), (n, 2n), (2n+1, 4n), (4n+1, 8n) \ldots, (2^k, t)$. We observe that $\sum_{r=0}^{n-1}{t \choose r}^2$ is at most $(1/2^k)\sum_{r=0}^{t}{t \choose r}^2$ because $t \choose r$ monotonically increases for $r < t/2$.

Now, we can generalize for $t \in (2^k n, 2^{k+1}n)$ for $k > 0$,
\begin{align*}
\sum_{t=2^kn+1}^{2^{k+1}n} \dfrac{n}{2^{2t}}\sum^{n-1}_{r = 0}{t \choose r}^2  &\leq C\sum_{t=2^kn+1}^{2^{k+1}n}\dfrac{n}{2^{2t}}\dfrac{2^{2t}}{2^k\sqrt{t}} \\
&\leq C n \sqrt{\dfrac{n}{2^k}}
\end{align*}
Now, summing for all $k = 1, 2, \ldots$ we get that 
\begin{align*}
\sum_{k=1}^{\infty}\sum_{t=2^kn+1}^{2^{k+1}n} \dfrac{n}{2^{2t}}\sum^{n-1}_{r = 0}{t \choose r}^2  &\leq \sum_{k=1}^{\infty}C n \sqrt{\dfrac{n}{2^k}}
\end{align*}
$\text{trace}(P(J(1))) = O(n \sqrt{n}) = o(n^2)$
\begin{figure}[h]
	\centering
	\includegraphics[width=\columnwidth]{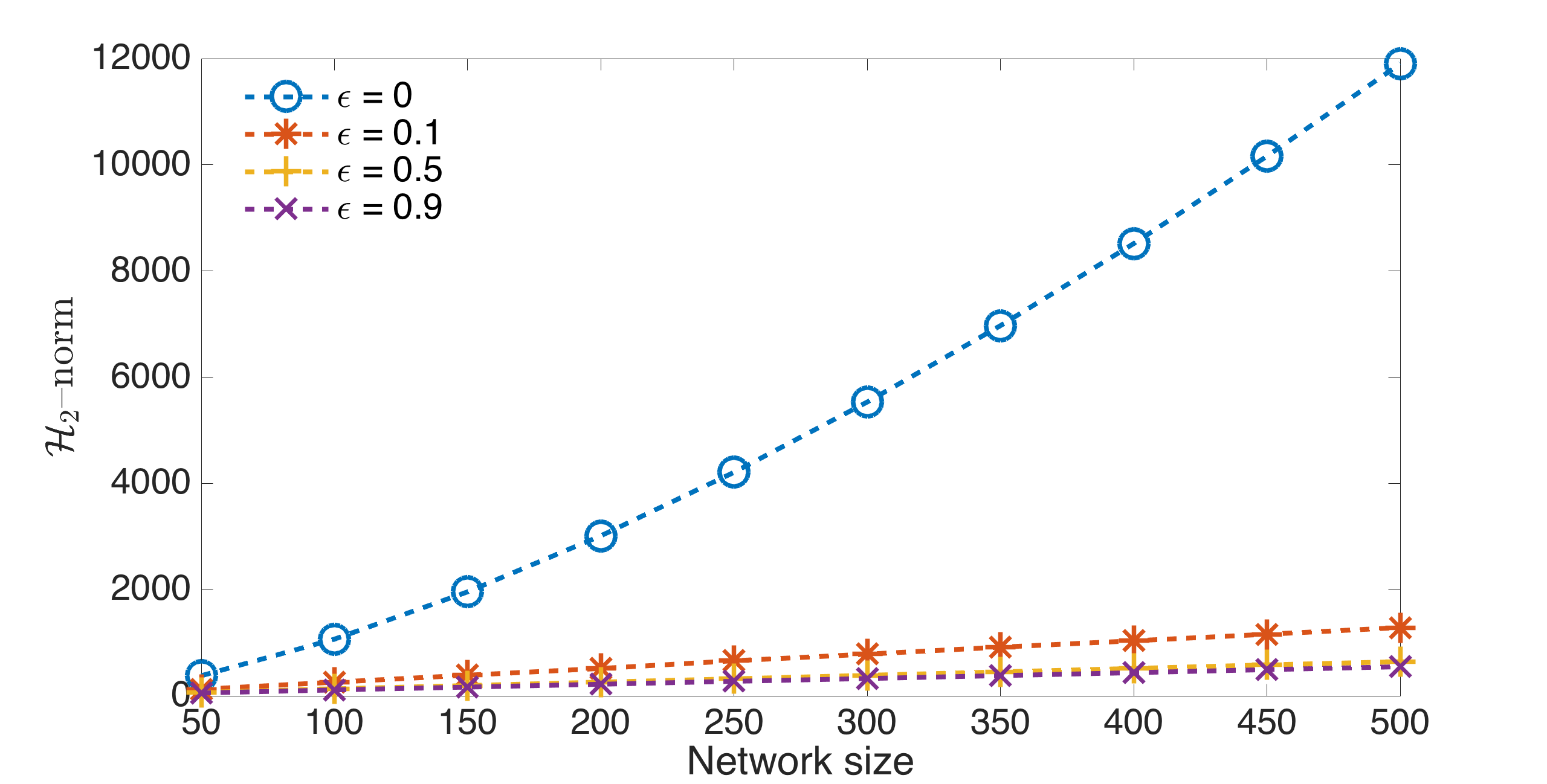}
	\caption{$\Hcal_2(I - K_{\text{asm}})$ variation with different $\epsilon$--balancing}
	\label{asm_controller}
\end{figure}
\subsection{Proof of Theorem~\ref{daron_equivalence}}
Observe that after a shock at time $k=0$ and with $x(0) = \textbf{0}$. We have that $x(k) = A^k \omega$. Then
\[
x_{\infty} = \sum_{k=0}^{\infty}\textbf{1}^T A^{k} \omega
\]
Since $\rho(A) < 1$, we have that $x_{\infty} = \textbf{1}^T (I - A)^{-1}\omega = \sum_{i=1}^n c_i \omega_i$, where $c_i$ are the column sums of $(I-A)^{-1}$. If $c_{\max} = \Theta(1)$ then it is obvious that, by Large Deviation Principle (or even a simple Chernoff Bound), would give us 
\begin{align*}
-&\lim_{n \rightarrow \infty}(1/n)\log{(\Pbb(|x_{\infty}| > n z))} < \infty 
\end{align*}
On the other hand if $c_{\max} = \Omega(g(n))$ where $g(n) \rightarrow \infty$ with $n$. Then it is true that 
\begin{align*}
&\Pbb(|x_{\infty}| > n z) > \Pbb(c_{\max}x_{1} > n z) \\
& \Pbb(|x_{\infty}| > n z) > \exp{(-n z/c_{\max})}\\
&\Longrightarrow -\lim_{n \rightarrow \infty}(1/n)\log{(\Pbb(|x_{\infty}| > n z))} > \infty
\end{align*}
%
\begin{lemma}
	\label{equivalenceLyap}
	Under Assumption~\ref{economic_assumption}, we have that 
	\[
		||(I-A_n)^{-1}||_1 =O(1) \Longleftrightarrow ||P(A_n/\sqrt{\rho(A_n)})||_1 = O(1)
	\]
\end{lemma}
\begin{proof}
By the assumptions on network matrix and from Eq. 32 in Ostrowski et. al.~\cite{ostrowski1960eigenvector} that 
\begin{equation}
\label{ostr}
\dfrac{R^{(k)}}{r^{(k)}} \leq \gamma
\end{equation}
where $\gamma = \frac{v_{\max}}{v_{\min}}$, $Av = \lambda_{PF} v$ and for every $k$, where $R^{(k)}, r^{(k)}$ are the largest and smallest row sum respectively of $A^{k}$. Now by Assumption~\ref{economic_assumption}, we have that $\gamma = O(1)$. Then for every $k$ we have 
\[
\dfrac{R^{(k)}}{r^{(k)}} = \Theta(1)
\]
Further observe that the matrix $A^k/\lambda_{PF}^k$ has Perron root $1$. Then from the Perron-Frobenius theorem for non--negative matrices it follows that 
$$\hat{r}^{(k)} \leq 1 \leq \hat{R}^{(k)}$$
where $\hat{R}^{(k)}, \hat{r}^{(k)}$ are the largest and smallest row sum respectively of $A^{k}/\lambda^k_{PF}$ and from the preceding result in Eq.~\ref{ostr} it follows that any row sum of $A^k/\lambda_{PF}^k$ is $\Theta(1)$, \textit{i.e.}, does not decay with $k$. Thus we have that 
\[
\hat{r}^{(k)} ||A^k||_1\leq \Bigg|\Bigg|\Bigg(\dfrac{A^T}{\lambda_{PF}}\Bigg)^kA^k\Bigg|\Bigg|_1 \leq \hat{R}^{(k)} ||A^k||_1
\]
Since all matrices are non--negative we have that
\begin{align*}
\sum_{k=0}^{\infty}\Bigg|\Bigg|\Bigg(\dfrac{A^T}{\lambda_{PF}}\Bigg)^kA^k\Bigg|\Bigg|_1 &=  \sum_{k=0}^{\infty}\Theta(||A^k||_1) \\
\Bigg|\Bigg|\sum_{k=0}^{\infty}\Bigg(\dfrac{A^T}{\sqrt{\lambda_{PF}}}\Bigg)^k\Bigg(\dfrac{A}{\sqrt{\lambda_{PF}}}\Bigg)^k\Bigg|\Bigg|_1 &= \Theta(||(I-A)^{-1}||_1) \\
||P(A/\sqrt{\lambda_{PF}})||_1 &= \Theta(||(I-A)^{-1}||_1)
\end{align*}
\end{proof}
\subsection{Proof of Theorem~\ref{daron}}
From Eqn.~\eqref{DT_LTI}, we have $x(k) = A^k \epsilon$ for all $k$. Then, $\sum_{k=0}^{\infty} x(k) = (I - A)^{-1}\epsilon$. From Proposition 5 in~\cite{alireza_daron} we have that a large network exhibits tail risk with respect to the standard normal distribution (or macroeconomic tail risk) iff $\lim_{n \rightarrow \infty}||v||_{\infty}\sqrt{n}/||v||_2 = \infty$, where $v = \textbf{1}^T (I - A)^{-1}$. From Assumption~\ref{economic_assumption} we have that $||(I-A)^{-1}||_{\infty} = \Theta(1)$. Then this means that the sum of all elements of $(I-A)^{-1}$ is $\Theta(n)$ (since the ratio of maximum row sum and minimum row sum is of constant order). The condition 
\[
\lim_{n \rightarrow \infty}||v||_{\infty}\sqrt{n}/||v||_2 < \infty
\]
implies that all column sums are of the same order (in $n$) which means that $||(I-A_n)^{-1}||_1 = \Theta(1)$ (because sum of column sum must equal sum of row sum). Thus 
\[
\lim_{n \rightarrow \infty}||v||_{\infty}\sqrt{n}/||v||_2 < \infty \Longrightarrow  ||(I-A_n)^{-1}||_1 = \Theta(1)
\]
The other direction is trivial, implying that 
\[
\lim_{n \rightarrow \infty}||v||_{\infty}\sqrt{n}/||v||_2 < \infty \Longleftrightarrow  ||(I-A_n)^{-1}||_1 = \Theta(1)
\]
\subsection{Proof of Proposition~\ref{macro}}
For $R_n, C_n$ we have that $\pi^T = (1/n)\textbf{1}^T$. Since $\pi_{\max}/\pi_{\min} = 1$, we have that $||(I - A_n)^{-1}||_1 = \Theta(1)$ for both of these. 
\begin{align*}
||\pi^T(I-A_n)^{-1}||_{\infty} &= ||\pi^T||_{\infty}(1 - \lambda_{PF})^{-1} \\
\pi_{\min}||(I-A_n)^{-T}||_{\infty} &\leq \pi_{\max}(1 - \lambda_{PF})^{-1}
\end{align*}
For $S_n$ we have that $||S_n||_1 = \Theta(n)$.
\subsection{Proof of Corollary~\ref{centrality}}
\label{centrality_proof}
Observe that
\begin{align*}
||A^T \pi||_{\infty} &= \lambda_{PF} ||\pi||_{\infty} \\
||A^T||_{\infty} \pi_{\min} &\leq \lambda_{PF} \pi_{\max} \\
||A||_{1} &\leq \lambda_{PF} \frac{\pi_{\max}}{\pi_{\min}}
\end{align*}
Then 
\begin{align*}
||(I - A)^{-1}||_1 &= \sum_{k=0}^{\infty} ||A^k||_1 \\
&\leq \sum_{k=0}^{\infty} \lambda^k_{PF} \frac{\pi_{\max}}{\pi_{\min}} \\
&= \frac{\pi_{\max}}{\pi_{\min}}  \frac{1}{1 - \lambda_{PF}}
\end{align*}
\subsection{Economic Model}
\label{economic_setup}

Consider  the economy consisting of $n$ competitive sectors denoted by $\{1, 2, \ldots, n\}$, each producing a distinct product. Each sector corresponds to a node in the network graph. Firms in each sector employ Cobb-Douglas production technologies with constant returns to scale. Formally, for each sector, $i$, we have
\begin{align}
\label{cobb_douglas}
\log{(x_{i, t+1})} &= \log{(\Sigma_{i, t})} + \log{(\eta_{i, t})} + \nonumber \\
& \mu\Bigg(\sum_{j=1}^n a_{ij} \log{(y_{ij, t})}\Bigg) + (1 - \mu)\log{(l_{i, t})}
\end{align} 
At, each time, $t$, $x_{i, t}$ is the output of sector $i$, $\Sigma_{i, t} \geq 0$ (since output of the network is always nonnegative) is the total factor productivity, $l_{i, t}$ is labor input to sector $i$, $y_{ij, t}$ is amount of output of sector $j$ used for the production of output of sector $i$, and $\eta_{i, t} > 0$ is some normalization constant. A larger $a_{ij}$ means that sector $j$ is more important in the production of output of sector $i$. Constant returns to scale implies $\sum_{j=1}^n a_{ij} = 1$ for all $i$, where $a_{ij} \geq 0$. From now, $A_n = \lbrack a_{ij} \rbrack$ will be referred as the economy's input-output, or network, matrix, where the network is denoted by $\Ncal(A_n; \Gcal)$, and $\Gcal$ is the graph induced by $A_n$. In Eq.~\eqref{cobb_douglas}, $\Sigma_{i, t}$ is a multiplicative production factor in the dynamics of the network, where under no shock, $\Sigma_{i, t} = 1$. We are interested in the case when $\Sigma_{i, 0} < 1$ and $\Sigma_{i, t} = 1$ for all $t > 0$ and $i \in \{1, 2, \ldots, n\}$. This corresponds to a negative shock of the form $\log{(\Sigma_{i, 0})}$ in the notation of Eq.~\eqref{DT_LTI}. We further assume that $\log{(\Sigma_{i, 0})} \sim \Ecal(-m^2, 1)$. The mathematical formulation in~\eqref{cobb_douglas} is completed by two resource constraints:
\begin{align}
x_{i,t} &= \sum_{j = 1}^n y_{ij, t} + c_{i, t} \nonumber\\
H &= \sum_{n=1}^t l_{i,t} \label{constraints} + Z_t
	\end{align}
Here $Z_t$ is the total leisure, $c_{i, t}$ is the consumption of product $i$ at node $i$ at time $t$, $c_t = \sum_{i=1}^n c_{i, t}$  and the sum of total input labor and leisure is a constant. Finally, we assume that economy maximizes the following preference function over consumption and leisure (See Sections 2, 3 in~\cite{long1983real}): 
\begin{align}
\label{pref_function}
V_0 &= \text{max}_{c, Z}\,\Ebb_{\Sigma_0}\Bigg[\sum_{s=0}^{\infty} u(c_s, Z_s) \Bigg] \nonumber\\
u(c_t, Z_t) &= \beta_0 Z_t + \sum_{i = 1}^n \beta_i \log{(c_{i, t})} 
\end{align}
where $\beta_i \geq 0$ but $\beta_0 = 0$, and $\sum_{j=1}^n\beta_j = 1$; $\beta_i$ is $i$'s share in the household's utility function. There are two forms of heterogeneity - primitive and network (See~\cite{alireza_daron}). Intuitively, primitive heterogeneity stems from the difference in preferences, \textit{i.e.}, $\beta_i$, across different sectors, meanwhile network heterogeneity is due to interconnections of input-output matrix. Since in this work we are concerned with topological dependence only, we will impose the following additional assumption.
\begin{assumption}
	\label{primitive_assumption}
	The utility function has no sectoral preferences, \textit{i.e.}, $\beta_i = 1/n$.
\end{assumption}

Under optimal consumption and labor, \textit{i.e.}, optimizing $V_0$ under the constraints Eqns.~\eqref{cobb_douglas},~\eqref{constraints}, we have Eq. 20 of~\cite{long1983real},
\begin{equation}
\label{dynamic_long}
\hat{x}_{t+1} = \mu A_n \hat{x}_{t} + \eta + \epsilon_{t}
\end{equation} 
Here $\epsilon_{i, t} = \log{(\Sigma_{i,t})}$, $\hat{x}_t = \lbrack \log{(x_{1, t})}, \ldots, \log{(x_{n, t})}\rbrack^T$, and $\eta$ is some constant input to the network. Now, we assume that the economic network is hit with a shock, $\epsilon_t = \epsilon$, at $t = 0$ ($\epsilon_{t} = 0$ for all $t \geq 1$);

Under no shock, \textit{i.e.}, $\epsilon_{t} = 0$ for all $t \geq 0$, we have in Eq.~\eqref{dynamic_long} that
\begin{equation*}
\hat{x}^{\text{ns}}_{t+1} = \mu A_n \hat{x}^{\text{ns}}_{t} + \eta
\end{equation*}

This gives us,
\begin{equation}
\label{reduced_form}
\hat{x}_{t+1} - \hat{x}^{\text{ns}}_{t+1} = \mu A_n (\hat{x}_{t} - \hat{x}^{\text{ns}}_{t}) + \epsilon_t
\end{equation}

The quantity of interest here is $y_t = \hat{x}_{t} - \hat{x}^{\text{ns}}_{t}$, \textit{i.e.}, the output deficit in presence of a shock. Then we have similar to Eq.~\eqref{deficit_eq}
\[
y_{t+1} = \mu A_n y_t + \epsilon_t
\]
Then in the notation of Eq.~\eqref{deficit_eq} $\omega_t = \epsilon_t$ and $A_{379} = \mu A_n$.

\begin{proposition}
	\label{combinatorial}
	For any $n$ and $p \leq n$ we have
	\[
	{2n \choose 2p} \geq {n \choose p}^2
	\]
\end{proposition}
\begin{proof}
	Consider $(1+x)^{2n}$, then the coefficient of $x^{2p}$ is given by 
	\[
	{2n \choose 2p}
	\]
	Alternately,
	\begin{align*}
	(1+x)^{2n} &= (1+x)^n (1+x)^n
 	\end{align*}
 	The coefficient of $x^p$ in $(1+x)^n$ is ${n \choose p}$. Since $x^{2p}$ can form in many ways $x^p x^p, x^{p-1} x^{p+1}$. It is clear that 
 	\[
 	{2n \choose 2p} \geq {n \choose p}^2
 	\]
 	where ${n \choose p}^2$ by choosing $x^p$ from both $(1+x)^n$s.
\end{proof}

\end{document}